\documentclass
[
    twoside,                 
    openright,               
    cleardoublepage = empty, 
    fontsize = 12 pt,        
    american,                
    captions = tableheading, 
    numbers = noenddot,      
    footheight = 35 pt,      
    colorinlistoftodos
]
{scrbook}

\newif\ifprintVersion   
\newif\ifprofessionalPrint 
\newif\iffancyTheorems  
\newif\ifboldNumberSets 
\newif\ifbachelorThesis 

\printVersionfalse
\professionalPrintfalse
\fancyTheoremstrue
\boldNumberSetstrue
\bachelorThesistrue

\newcommand*{\printTitle}{}
\newcommand*{\printGermanTitle}{}
\newcommand*{\myTitle}[2]{\renewcommand*{\printTitle}{#1}\renewcommand*{\printGermanTitle}{#2}}
\newcommand*{\printTitleBold}{\textbf{\printTitle}}

\newcommand*{\printAuthor}{}
\newcommand*{\myName}[1]{\renewcommand*{\printAuthor}{#1}}

\newcommand*{\printProgram}{}
\newcommand*{\myProgram}[1]{\renewcommand*{\printProgram}{#1}}

\newcommand*{\printDateReceived}{}
\newcommand*{\dateOfHandingIn}[1]{\renewcommand*{\printDateReceived}{#1}}

\newcommand*{\printSubject}{}
\newcommand*{\mySubject}[1]{\renewcommand*{\printSubject}{#1}}

\newcommand*{\printKeywords}{}
\newcommand*{\myKeywords}[1]{\renewcommand*{\printKeywords}{#1}}

\newcommand*{\printNameOfSupervisor}{}
\newcommand*{\nameOfMySupervisor}[1]{\renewcommand*{\printNameOfSupervisor}{#1}}

\newcommand*{\printAdditionalExaminers}{}
\newcommand*{\additionalExaminers}[1]{\renewcommand*{\printAdditionalExaminers}{#1}}

\newlength{\extraborderlength}
\newcommand*{\extraBorder}[1]{\setlength{\extraborderlength}{#1}}

\newlength{\mybindingcorrection}
\newcommand*{\bindingCorrection}[1]{\setlength{\mybindingcorrection}{#1}} 

    \extraBorder{3 mm}

    \bindingCorrection{6 mm}

    \myTitle{Conquering the Multiverse: The River Voting Method with Efficient Parallel Universe Tiebreaking}{Bezwingen des Multiversums: Die River Wahlmethode mit Effizientem Parallel Universe Tiebreaking}

    \myName{Jannes Michel Malanowski}

    \myProgram{IT Systems Engineering}

    \dateOfHandingIn{01. November 2024}

    \nameOfMySupervisor{Prof.\,Dr. Tobias Friedrich}

    \additionalExaminers{Michelle Luise Döring \newline Stefan Neubert}

    \mySubject{Examintation of the River Social Choice Function with Parrallel Universe Tiebreaking}

    \myKeywords{bachelor thesis | world-changing | computational social choice | River | parallel universe tie breaking | polynomial time algorithm}

\usepackage[utf8]{inputenc} 
\usepackage[T1]{fontenc}    
\usepackage
[
    ngerman,         
    main = american, 
]
{babel}                     

\input glyphtounicode
\usepackage{calc} 

\newlength{\myparindent}
\newlength{\myparskip}
\setfootnoterule{0 cm}

\deffootnote[1.2 em]{1.2 em}{0 em}{\makebox[1.4 em][l]{\textbf{\thefootnotemark}}}

\makeatletter%
    \@removefromreset{footnote}{chapter}%
\makeatother

\usepackage[dvipsnames, table]{xcolor} 

\definecolor{stroke1}{HTML}{2574A9} 

\colorlet{captionlabel}{black}
\colorlet{footerpagenr}{black}
\colorlet{footerchapter}{stroke1}
\colorlet{footerchaptername}{black}
\colorlet{footersection}{stroke1}
\colorlet{footersectionname}{black}
\colorlet{chapternumber}{stroke1}

\newlength{\mypaperwidth}
\newlength{\mypaperheight}
\newlength{\mybodywidth}
\newlength{\mybodyheight}
\newlength{\myoutermargin}
\ifprintVersion
    \ifprofessionalPrint
    \else
    \fi
\else
\fi

\newlength{\mytopmargin}
\ifprintVersion
    \ifprofessionalPrint
    \fi
\fi

\newlength{\myinnermargin}
\ifprintVersion
    \ifprofessionalPrint
    \fi
\fi

\newlength{\mybottommargin}
\ifprintVersion
    \ifprofessionalPrint
    \fi
\fi

\newcommand{\goldenratio}{1.618}

\newlength{\myheadsep} 
\ifprintVersion
    \ifprofessionalPrint
    \fi
\fi

\newlength{\myfootskip} 
\ifprintVersion
    \ifprofessionalPrint
    \fi
\fi

\newlength{\mymargininnersep} 
\newlength{\mymarginoutersep} 
\ifprintVersion
    \ifprofessionalPrint
    \fi
\fi

\newlength{\mymarginwidth} 
\newlength{\mymarginwidthwithinnersep} 
\usepackage
[
    \ifprintVersion
        \ifprofessionalPrint
            paperwidth = \mypaperwidth + 2\extraborderlength + \mybindingcorrection,
            paperheight = \mypaperheight + 2\extraborderlength,
        \else
            paperwidth = \mypaperwidth,
            paperheight = \mypaperheight,
        \fi
    \else
        paperwidth = \mypaperwidth,
        paperheight = \mypaperheight,
    \fi
    textwidth = \mybodywidth,
    textheight = \mybodyheight,
    outer = \myoutermargin,
    top = \mytopmargin,
    headsep = \myheadsep,
    footskip = \myfootskip,
    marginparsep = \mymargininnersep,
    marginparwidth = \mymarginwidth,
]
{geometry} 

\usepackage
[
]
{scrlayer-scrpage} 

\clearpairofpagestyles

\KOMAoptions
{%
    headwidth = \textwidth + \mymarginwidthwithinnersep,%
    footwidth = \myoutermargin : \textwidth,%
}

\automark[chapter]{chapter}
\automark*[section]{}

\lehead%
{%
    \begin{minipage}[b]{\mymarginwidth}%
        \small\raggedleft\normalfont\textsf{\textbf{\color{footerchapter}\chaptername\ \thechapter}}
    \end{minipage}
}
\cehead{\hspace*{\mymarginwidthwithinnersep}\parbox{\textwidth}{\raggedright\leftmark}}

\rohead%
{%
    \Ifstr{\rightmark}{\leftmark}%
    {%
        \begin{minipage}[b]{\mymarginwidth}%
            \small\raggedright\normalfont\textsf{\textbf{\color{footersection}Chapter\ \thechapter}}%
        \end{minipage}%
    }%
    {%
        \begin{minipage}[b]{\mymarginwidth}%
            \small\raggedright\normalfont\textsf{\textbf{\color{footersection}Section\ \thesection}}%
        \end{minipage}%
    }%
}
\cohead{\hspace*{-\mymarginwidthwithinnersep}\parbox{\textwidth}{\raggedleft\rightmark}}

\lefoot*%
{%
    \vspace*{1 ex}%
    {\color{stroke1}\rule{\myoutermargin - \mymargininnersep}{0.5 mm}}\\
    \begin{minipage}[b]{\myoutermargin - \mymargininnersep}%
        \raggedleft\normalfont\color{footerpagenr}\textbf{\thepage}%
    \end{minipage}%
}
\rofoot*%
{%
    {\color{stroke1}\rule{\myoutermargin - \mymargininnersep}{0.5 mm}}\\
    \begin{minipage}[b]{\myoutermargin - \mymargininnersep}%
        \raggedright\normalfont\color{footerpagenr}\textbf{\thepage}%
    \end{minipage}%
}

\usepackage{caption}
\usepackage{tikz} 

\newlength{\mytmpa}
\newlength{\mytmpb}

\renewcommand*{\partlineswithprefixformat}[3]%
{%
    #2
    \thispagestyle{empty}
    \setlength{\mytmpa}{0.618\mypaperwidth}%
    \setlength{\mytmpb}{0.382\mypaperheight}%
    \ifprintVersion
        \ifprofessionalPrint
            \setlength{\mytmpa}{0.618\mypaperwidth + \mybindingcorrection + \extraborderlength}%
            \setlength{\mytmpb}{0.382\mypaperheight + \extraborderlength}%
        \fi
    \fi
    \begin{tikzpicture}[overlay, remember picture]%
        \node [inner sep = 0, outer sep = 0, anchor = north] at (current page.north west)%
        {%
            \begin{tikzpicture}[overlay, remember picture]%
            \draw[color = stroke1, line width = 0.7 mm] (\mytmpa, 0) -- (\mytmpa, -\mytmpb);%
            \end{tikzpicture}%
        };%
        \node (align) [align = right, below = \mytmpb - 2 ex, inner sep = 0, outer sep = 0, anchor = north west] at (current page.north west)%
        {%
            \hspace{\mytmpa}\hspace{0.5 em}\partname\ \thepart\\[1 ex]
            \color{stroke1}#3%
        };%
    \end{tikzpicture}%
}
\RedeclareSectionCommand%
[%
    font = \normalfont\Huge\sffamily,
    prefixfont = \normalfont\Huge\sffamily,
]
{part}

\usepackage{etoolbox}

\newbool{chapterHasANumber}
\newbool{chapterHasAStar}
\renewcommand*{\chapterlinesformat}[3]%
{%
    \Ifnumbered{#1}{\setbool{chapterHasANumber}{true}}{\setbool{chapterHasANumber}{false}}%
    \Ifstr{#2}{}{\setbool{chapterHasAStar}{true}}{\setbool{chapterHasAStar}{false}}%
    \ifboolexpr{bool{chapterHasANumber} and not bool{chapterHasAStar}}%
    {%
        \begin{tikzpicture}[overlay, remember picture]%
            \node [right = \myinnermargin, below = \mytopmargin, inner sep = 0, outer sep = 0, anchor = north west] (numbernode) at (current page.north west)%
            {%
                \hspace{\myinnermargin}%
                \sffamily\fontsize{60}{60}\selectfont%
                \color{chapternumber}%
                \thechapter%
            };%
            \node [inner sep = 0, outer sep = 0, anchor = north west] at (numbernode.south west)%
            {%
                \begin{tikzpicture}[overlay, remember picture]%
                    \draw[color = stroke1, line width = 0.7 mm] (\myinnermargin, -1 ex) -- (\paperwidth, -1 ex);%
                \end{tikzpicture}%
            };%
            \node (align) [text width = \textwidth - 2 cm, align = right, right = \myinnermargin + \mybodywidth, inner sep = 0, outer sep = 0, anchor = east] at (numbernode.west)%
            {%
                #3%
            };%
        \end{tikzpicture}%
    }%
    {%
        \begin{tikzpicture}[overlay, remember picture]%
            \node [right = \myinnermargin, below = \mytopmargin, inner sep = 0, outer sep = 0, anchor = north west] (numbernode) at (current page.north west)%
            {%
                \hspace{\myinnermargin}%
                \sffamily\fontsize{60}{60}\selectfont%
                \color{white}%
                \thechapter%
            };%
            \node [inner sep = 0, outer sep = 0, anchor = north west] at (numbernode.south west)%
            {%
                \begin{tikzpicture}[overlay, remember picture]%
                    \draw[color = stroke1, line width = 0.7 mm] (\myinnermargin, -1 ex) -- (\paperwidth, -1 ex);%
                \end{tikzpicture}%
            };%
            \node (align) [align = left, right = \myinnermargin, inner sep = 0, outer sep = 0, anchor = south west] at (numbernode.south west)%
            {%
                #3%
            };%
        \end{tikzpicture}%
    }%
}
\RedeclareSectionCommand%
[%
    font = \color{stroke1}\normalfont\huge\sffamily,
    afterskip = 20 pt,
]
{chapter}

\BeforeStartingTOC[toc]{\pagestyle{plain}}
\AfterStartingTOC{\thispagestyle{plain}}
\usepackage
[
    sortcites,              
    style = alphabetic,     
    defernumbers,           
    safeinputenc,           
    backref = true,         
    backrefstyle = three,   
    hyperref = true,        
    maxbibnames = 99,       
    maxcitenames = 2,       
]
{biblatex} 

\renewbibmacro{in:}%
{%
    \ifentrytype{article}{}{\printtext{\bibstring{in}\intitlepunct}}%
}

\renewbibmacro*{volume+number+eid}%
{%
    \printfield{volume}%
    \iffieldundef{number}{}{\addcolon}%
    \printfield{number}%
    \setunit*{\addcomma\space}%
    \printfield{eid}%
}

\DefineBibliographyStrings{english}%
{%
    backrefpage  = {\lowercase{s}ee page}, 
    backrefpages = {\lowercase{s}ee pages} 
}

\DeclareFieldFormat[article]{title}{\textbf{\color{stroke1}#1}}
\DeclareFieldFormat[inproceedings]{title}{\textbf{\color{stroke1}#1}}
\DeclareFieldFormat[thesis]{title}{\textbf{\color{stroke1}#1}}
\DeclareFieldFormat[book]{title}{\textbf{\color{stroke1}#1}}
\DeclareFieldFormat[unpublished]{title}{\textbf{\color{stroke1}#1}}
\DeclareFieldFormat[report]{title}{\textbf{\color{stroke1}#1}}
\DeclareFieldFormat[inbook]{chapter}{\textbf{\color{stroke1}#1}}
\DeclareFieldFormat[inbook]{title}{#1}
\DeclareFieldFormat{pages}{#1}

\newtoggle{authorend}
\togglefalse{authorend}

\DeclareBibliographyDriver{article}%
{%
  \usebibmacro{bibindex}%
  \usebibmacro{begentry}%
  \iftoggle{authorend}{}{\usebibmacro{author/translator+others}}%
  \setunit{\labelnamepunct}\newblock
  \usebibmacro{title}%
  \newunit
  \printlist{language}%
  \newunit\newblock
  \usebibmacro{byauthor}%
  \newunit\newblock
  \usebibmacro{bytranslator+others}%
  \newunit\newblock
  \printfield{version}%
  \newunit\newblock
  \usebibmacro{in:}%
  \usebibmacro{journal+issuetitle}%
  \newunit
  \usebibmacro{byeditor+others}%
  \newunit
  \usebibmacro{note+pages}%
  \newunit\newblock
  \iftoggle{bbx:isbn}
  {\printfield{issn}}
  {}%
  \newunit\newblock
  \usebibmacro{doi+eprint+url}%
  \newunit\newblock
  \usebibmacro{addendum+pubstate}%
  \setunit{\bibpagerefpunct}\newblock
  \usebibmacro{pageref}%
  \newunit\newblock
  \iftoggle{bbx:related}
  {\usebibmacro{related:init}%
    \usebibmacro{related}}
  {}%
  \usebibmacro{finentry}%
  \iftoggle{authorend}{\usebibmacro{author/translator+others}}{}%
}

\DeclareBibliographyDriver{inbook}%
{%
  \usebibmacro{bibindex}%
  \usebibmacro{begentry}%
  \iftoggle{authorend}{}{\usebibmacro{author/translator+others}}%
  \setunit{\labelnamepunct}\newblock
  \usebibmacro{chapter+pages}%
  \newunit
  \printlist{language}%
  \newunit\newblock
  \usebibmacro{byauthor}%
  \newunit\newblock
  \usebibmacro{in:}%
  \usebibmacro{bybookauthor}%
  \newunit\newblock
  \usebibmacro{maintitle+booktitle}%
  \newunit\newblock
  \usebibmacro{byeditor+others}%
  \newunit\newblock
  \printfield{edition}%
  \newunit
  \iffieldundef{maintitle}
  {\printfield{volume}%
    \printfield{part}}
  {}%
  \newunit
  \printfield{volumes}%
  \newunit\newblock
  \usebibmacro{series+number}%
  \newunit\newblock
  \printfield{note}%
  \newunit\newblock
  \usebibmacro{publisher+location+date}%
  \newunit\newblock
  \newunit\newblock
  \iftoggle{bbx:isbn}
  {\printfield{isbn}}
  {}%
  \newunit\newblock
  \usebibmacro{doi+eprint+url}%
  \newunit\newblock
  \usebibmacro{addendum+pubstate}%
  \setunit{\bibpagerefpunct}\newblock
  \usebibmacro{pageref}%
  \newunit\newblock
  \iftoggle{bbx:related}
  {\usebibmacro{related:init}%
    \usebibmacro{related}}
  {}%
  \usebibmacro{finentry}%
  \iftoggle{authorend}{\usebibmacro{author/translator+others}}{}%
}

\DeclareBibliographyDriver{inproceedings}%
{%
  \usebibmacro{bibindex}%
  \usebibmacro{begentry}%
  \iftoggle{authorend}{}{\usebibmacro{author/translator+others}}%
  \setunit{\labelnamepunct}\newblock
  \usebibmacro{title}%
  \newunit
  \printlist{language}%
  \newunit\newblock
  \usebibmacro{byauthor}%
  \newunit\newblock
  \usebibmacro{in:}%
  \usebibmacro{maintitle+booktitle}%
  \newunit\newblock
  \usebibmacro{event+venue+date}%
  \newunit\newblock
  \usebibmacro{byeditor+others}%
  \newunit\newblock
  \iffieldundef{maintitle}
  {\printfield{volume}%
    \printfield{part}}
  {}%
  \newunit
  \printfield{volumes}%
  \newunit\newblock
  \usebibmacro{series+number}%
  \newunit\newblock
  \printfield{note}%
  \newunit\newblock
  \printlist{organization}%
  \newunit
  \usebibmacro{publisher+location+date}%
  \newunit\newblock
  \usebibmacro{chapter+pages}%
  \newunit\newblock
  \iftoggle{bbx:isbn}
  {\printfield{isbn}}
  {}%
  \newunit\newblock
  \usebibmacro{doi+eprint+url}%
  \newunit\newblock
  \usebibmacro{addendum+pubstate}%
  \setunit{\bibpagerefpunct}\newblock
  \usebibmacro{pageref}%
  \newunit\newblock
  \iftoggle{bbx:related}
  {\usebibmacro{related:init}%
    \usebibmacro{related}}
  {}%
  \usebibmacro{finentry}%
  \iftoggle{authorend}{\usebibmacro{author/translator+others}}{}%
}

\DeclareBibliographyDriver{thesis}%
{%
  \usebibmacro{bibindex}%
  \usebibmacro{begentry}%
  \iftoggle{authorend}{}{\usebibmacro{author}}%
  \setunit{\labelnamepunct}\newblock
  \usebibmacro{title}%
  \newunit
  \printlist{language}%
  \newunit\newblock
  \usebibmacro{byauthor}%
  \newunit\newblock
  \printfield{note}%
  \newunit\newblock
  \printfield{type}%
  \newunit
  \usebibmacro{institution+location+date}%
  \newunit\newblock
  \usebibmacro{chapter+pages}%
  \newunit
  \printfield{pagetotal}%
  \newunit\newblock
  \iftoggle{bbx:isbn}
  {\printfield{isbn}}
  {}%
  \newunit\newblock
  \usebibmacro{doi+eprint+url}%
  \newunit\newblock
  \usebibmacro{addendum+pubstate}%
  \setunit{\bibpagerefpunct}\newblock
  \usebibmacro{pageref}%
  \newunit\newblock
  \iftoggle{bbx:related}
  {\usebibmacro{related:init}%
    \usebibmacro{related}}
  {}%
  \usebibmacro{finentry}%
  \iftoggle{authorend}{\usebibmacro{author}}{}%
}

\providebool{bbx:subentry}
\newbibmacro*{citenum}%
{
  \printtext[bibhyperref]{
    \printfield{prefixnumber}%
    \printfield{labelnumber}%
    \ifbool{bbx:subentry}
    {\printfield{entrysetcount}}
    {}}%
}

\DeclareCiteCommand{\conline}[\mkbibbrackets]
{\usebibmacro{prenote}}
{\usebibmacro{citeindex}%
  \usebibmacro{citenum}}
{\multicitedelim}
{\usebibmacro{postnote}}       
\usepackage
[
    babel = true, 
]
{microtype}           
\usepackage{csquotes} 

\usepackage{amsmath}
\usepackage{amssymb}
\usepackage{amsthm}
\usepackage{thmtools}
\usepackage{mathtools}
\usepackage{thm-restate}
\usepackage{dsfont}        
\usepackage{braceMnSymbol} 

\usepackage
[
    ttscale = 0.85, 
]
{libertine} 
\usepackage
[
    libertine,    
    slantedGreek, 
    vvarbb,       
    libaltvw,     
]
{newtxmath} 
\usepackage{url} 
\usepackage{bm}  

\usepackage{graphicx} 
\usepackage
[
    subrefformat = simple, 
    labelformat = simple,  
]
{subcaption}         
\usepackage{wrapfig} 

\columnsep = \mymargininnersep

\usepackage{array}     
\usepackage{booktabs}  
\usepackage{longtable} 
\usepackage{pdflscape} 

\usepackage{enumitem} 

\usepackage
[
    ruled,         
    vlined,        
    linesnumbered, 
]
{algorithm2e} 

\usepackage{xspace}   
\usepackage
[
    shortcuts, 
]
{extdash}             
\usepackage{setspace} 

\usepackage{xparse}    
\usepackage{footnote}  
\usepackage{afterpage} 

\usepackage[normalem]{ulem}
\usepackage{makecell} 

\usepackage{siunitx}

\usetikzlibrary{arrows.meta}

\usepackage
[
    bookmarks = true,                 
    bookmarksopen = false,            
    bookmarksnumbered = true,         
    pdfstartpage = 1,                 
    pdftitle = {{\printTitle}},       
    pdfauthor = {{\printAuthor}},     
    pdfsubject = {{\printSubject}},   
    pdfkeywords = {{\printKeywords}}, 
    breaklinks = true,                
    \ifprintVersion
        hidelinks,                    
    \else
    colorlinks = true,            
    allcolors = stroke1,          
    \fi
]
{hyperref} 

\usepackage
[
    noabbrev,   
    nameinlink, 
]
{cleveref} 
\newcommand*{\colloquialDegreeName}{Master}
\newcommand*{\colloquialDegreeNameLowercase}{master}

\newcommand*{\degreeAbbreviation}{M.}

\ifbachelorThesis
    \renewcommand*{\colloquialDegreeName}{Bachelor}
    \renewcommand*{\colloquialDegreeNameLowercase}{bachelor}
    \renewcommand*{\degreeAbbreviation}{B.}
\fi

\makeatletter
    \def\IfEmptyTF#1%
    {%
        \if\relax\detokenize{#1}\relax%
            \expandafter\@firstoftwo%
        \else%
            \expandafter\@secondoftwo%
        \fi%
    }
\makeatother

\NewDocumentCommand{\mathOrText}{m}
{%
    \ensuremath{#1}\xspace%
}

\let\originalleft\left
\let\originalright\right
\renewcommand{\left}{\mathopen{}\mathclose\bgroup\originalleft}
\renewcommand{\right}{\aftergroup\egroup\originalright}

\makeatletter
    \DeclareRobustCommand{\bfseries}%
    {%
        \not@math@alphabet\bfseries\mathbf%
        \fontseries\bfdefault\selectfont%
        \boldmath%
    }
\makeatother

\xspaceaddexceptions{]\}}

\allowdisplaybreaks

\crefname{ineq}{inequality}{inequalities}
\creflabelformat{ineq}{#2{\upshape(#1)}#3} 

\crefname{term}{term}{terms}
\creflabelformat{term}{#2{\upshape(#1)}#3}

\let\oldfootnote\footnote

\newlength{\spaceBeforeFootnote} 
\newlength{\spaceAfterFootnote}  

\RenewDocumentCommand{\footnote}{o o o m}%
{%
    \IfNoValueTF{#1}%
    {%
        \oldfootnote{#4}%
    }%
    {%
        \setlength{\spaceBeforeFootnote}{\IfEmptyTF{#1}{0}{#1} em}%
        \IfNoValueTF{#2}%
        {%
            \hspace*{\spaceBeforeFootnote}\oldfootnote{#4}%
        }%
        {%
            \setlength{\spaceAfterFootnote}{\IfEmptyTF{#2}{0}{#2} em}%
            \hspace*{\spaceBeforeFootnote}\IfNoValueTF{#3}{\oldfootnote{#4}}{\oldfootnote[#3]{#4}}\hspace*{\spaceAfterFootnote}%
        }%
    }%
}

\makesavenoteenv{figure}
\makesavenoteenv{table}
\makesavenoteenv{tabular}

\iffancyTheorems
    \declaretheoremstyle
    [
        spaceabove = \topsep,
        spacebelow = \topsep,
        headfont = \bfseries,
        headformat = \textcolor{stroke1}{$\blacktriangleright$} \NAME~\NUMBER \NOTE,
        notefont = \bfseries,
        notebraces = {(}{)},
        bodyfont = \normalfont,
        postheadspace = 0.5 em,
        qed = \textcolor{stroke1}{\bfseries$\blacktriangleleft$},
    ]
    {myTheoremStyle}
    
    \declaretheorem
    [
        style = myTheoremStyle,
        name = Conjecture,
        numberwithin = section,
    ]
    {conjecture}
    \declaretheorem
    [
        style = myTheoremStyle,
        name = Proposition,
        sharenumber = conjecture,
    ]
    {proposition}
    \declaretheorem
    [
        style = myTheoremStyle,
        name = Claim,
    ]
    {claim}
    \declaretheorem
    [
        style = myTheoremStyle,
        name = Lemma,
        sharenumber = conjecture,
    ]
    {lemma}
    \declaretheorem
    [
        style = myTheoremStyle,
        name = Corollary,
        sharenumber = conjecture,
    ]
    {corollary}
    \declaretheorem
    [
        style = myTheoremStyle,
        name = Theorem,
        sharenumber = conjecture,
    ]
    {theorem}
    \declaretheorem
    [
        style = myTheoremStyle,
        name = Definition,
        sharenumber = conjecture,
    ]
    {definition}
    \declaretheorem
    [
        style = myTheoremStyle,
        name = Example,
        sharenumber = conjecture,
    ]
    {example}
    \declaretheorem
    [
        style = myTheoremStyle,
        name = Remark,
        sharenumber = conjecture,
    ]
    {remark}
        \declaretheorem
    [
        style = myTheoremStyle,
        name = Observation,
        sharenumber = conjecture,
    ]
    {observation}
\else
    \theoremstyle{plain}
    
    \newtheorem{conjecture}{Conjecture}[chapter]
    
    \newtheorem{claim}[conjecture]{Claim}
    \newtheorem{lemma}[conjecture]{Lemma}
    \newtheorem{corollary}[conjecture]{Corollary}
    \newtheorem{theorem}[conjecture]{Theorem}
    \newtheorem{definition}[conjecture]{Definition}

\fi

\NewDocumentCommand{\functionTemplate}{m m m m o}%
{%
    \IfNoValueTF{#5}%
    {%
        \mathOrText{#1\left#2{#4}\right#3}%
    }%
    {%
        \mathOrText{#1#5#2{#4}#5#3}%
    }%
}

\newcommand*{\leftBracketType}{(}
\newcommand*{\rightBracketType}{)}

\NewDocumentCommand{\createFunction}{m m o o}%
{%
    \renewcommand*{\leftBracketType}{\IfNoValueTF{#3}{(}{#3}}%
    \renewcommand*{\rightBracketType}{\IfNoValueTF{#4}{)}{#4}}%
    \NewDocumentCommand{#1}{o o}%
    {%
        \IfNoValueTF{##1}%
        {%
            \mathOrText{#2}%
        }%
        {%
            \functionTemplate{#2}{\leftBracketType}{\rightBracketType}{##1}[##2]%
        }%
    }%
}

\DeclareDocumentCommand{\probabilisticFunctionTemplate}{m m O{} o}
{%
    \functionTemplate{#1}%
    {\lbrack}%
    {\rbrack}%
    {#2\IfEmptyTF{#3}{}{\ \IfNoValueTF{#4}{\left}{#4}\vert\ \vphantom{#2}#3\IfNoValueTF{#4}{\right.}{}}}%
    [#4]%
}

\ifboldNumberSets

    \newcommand*{\indicatorFunctionSymbol}{\mathbf{1}}
\else

    \newcommand*{\indicatorFunctionSymbol}{\mathds{1}}
\fi

\RenewDocumentCommand{\Pr}{m O{} o}%
{%
    \probabilisticFunctionTemplate{\mathrm{Pr}}{#1}[#2][#3]%
}

\NewDocumentCommand{\E}{m O{} o}%
{%
    \probabilisticFunctionTemplate{\mathrm{E}}{#1}[#2][#3]%
}

\NewDocumentCommand{\Var}{m O{} o}%
{%
    \probabilisticFunctionTemplate{\mathrm{Var}}{#1}[#2][#3]%
}

\DeclareDocumentCommand{\bigO}{m o}%
{%
    \functionTemplate{\mathcal{O}}{(}{)}{#1}[#2]%
}

\DeclareDocumentCommand{\smallO}{m o}%
{%
    \functionTemplate{\mathrm{o}}{(}{)}{#1}[#2]%
}

\DeclareDocumentCommand{\bigTheta}{m o}%
{%
    \functionTemplate{\upTheta}{(}{)}{#1}[#2]%
}

\DeclareDocumentCommand{\bigOmega}{m o}%
{%
    \functionTemplate{\upOmega}{(}{)}{#1}[#2]%
}

\DeclareDocumentCommand{\smallOmega}{m o}%
{%
    \functionTemplate{\upomega}{(}{)}{#1}[#2]%
}

\DeclareDocumentCommand{\eulerE}{o}%
{%
    \mathOrText{\mathrm{e}\IfNoValueTF{#1}{}{^{#1}}}%
}

\DeclareDocumentCommand{\poly}{m o}%
{%
    \functionTemplate{\mathrm{poly}}{(}{)}{#1}[#2]%
}

\createFunction{\id}{\mathrm{id}}

\NewDocumentCommand{\ind}{m o o}%
{%
    \IfNoValueTF{#2}%
    {%
        \mathOrText{\indicatorFunctionSymbol_{#1}}%
    }%
    {%
        \functionTemplate{\indicatorFunctionSymbol_{#1}}{(}{)}{#2}[#3]%
    }%
}

\DeclareDocumentCommand{\dom}{m o}%
{%
    \functionTemplate{\mathrm{dom}}{(}{)}{#1}[#2]%
}

\DeclareDocumentCommand{\rng}{m o}%
{%
    \functionTemplate{\mathrm{rng}}{(}{)}{#1}[#2]%
}

\DeclareDocumentCommand{\d}{o}%
{%
    \mathrm{d}\IfNoValueTF{#1}{}{^{#1}}%
}

\DeclareDocumentCommand{\set}{m m o}%
{
    \mathOrText{\IfNoValueTF{#3}{\left}{#3}\{#1\ \IfNoValueTF{#3}{\left}{#3}\vert\
    \vphantom{#1}#2\IfNoValueTF{#3}{\right.}{}\IfNoValueTF{#3}{\right}{#3}\}}
}      

\preto\tabular{\setcounter{magicrownumbers}{0}}
\newcounter{magicrownumbers}

\newcommand{\sset}[1]{\{#1\}}
\newcommand{\msset}[1]{\{\,#1\,\}}
\newcommand{\abs}[1]{\left\lvert #1 \right\rvert}

\newcommand*{\Path}[2]{\mathOrText{{P}_{#1, #2}}}
\newcommand*{\PathP}{\mathOrText{{P}}}

\newcommand*{\alts}{\mathOrText{A}}
\newcommand*{\prefs}{\mathOrText{\mathbf{P}}}
\newcommand*{\allPrefs}{\mathOrText{\mathbb{P}}}

\newcommand{\mpref}[1]{\mathOrText{m_\prefs(#1)}}
\newcommand{\m}[1]{\mathOrText{m(#1)}}
\newcommand*{\mG}{\mathOrText{\mathcal{M}}}
\newcommand*{\mRV}{\mathOrText{\mathcal{M}^{RV}}}
\newcommand*{\mRVi}[1][i]{\mathOrText{\mathcal{M}^{RV}[ #1 ]}}

\newcommand*{\allDescOrders}{\mathOrText{\mathbb{O}_{\downarrow}}}

\newcommand*{\ine}{\mathOrText{\upGamma^{-}}}

\newcommand*{\msRV}{\mathOrText{\mathcal{M}^\textsc{sRV}\xspace}} 
\newcommand{\msRVi}[1][i]{\mathOrText{\mathcal{M}^\textsc{sRV}[ #1 ]\xspace}} 
\newcommand{\msRVm}[1][k]{\mathOrText{\mathcal{M}^\textsc{sRV}[m \geq #1]\xspace}} 
\newcommand{\EsRVm}[1][k]{\mathOrText{{E}^\textsc{sRV}[m \geq #1]\xspace}} 
\newcommand{\msRVgm}[1][k]{\mathOrText{\mathcal{M}^\textsc{sRV}[m > #1]\xspace}} 
\newcommand{\EsRVgm}[1][k]{\mathOrText{{E}^\textsc{sRV}[m > #1]\xspace}} 

\newcommand*{\wcand}{\mathOrText{w\xspace}}
\newcommand{\rspt}{\textormath{$T^{RSP}$\xspace}{T^{RSP}\xspace}}

\newcommand*{\rsptf}{\textit{recursively strongest path tree}\xspace}
\newcommand*{\ourPrim}{\textsc{directed Prim}\xspace}
\newcommand*{\Vrem}{\mathOrText{V \setminus S \xspace}}

\newcommand{\rvt}{\textormath{$T^{RV}(\ospecial)$\xspace}{T^{RV}(\ospecial)\xspace}}
\newcommand{\rvti}[1][i]{\textormath{$T^{RV}[#1](\ospecial)$\xspace}{T^{RV}[ #1 ](\ospecial)\xspace}}
\newcommand{\ospecial}{\mathOrText{{\lino}^{\textsc{RSP}}\xspace}}
\newcommand*{\lino}{\mathOrText{o}}
\newcommand*{\minweight}{\varrho}

\newcommand*{\eiin}{e_{\geq \text{in}}}

\newcommand*{\ebr}{e_{\text{break}}}
\newcommand*{\ebrm}{e_{\text{break}}} 

\DeclareMathOperator{\RV}{RV}
\DeclareMathOperator{\RVPUT}{RV-PUT}

\DeclareMathOperator{\str}{strength}

\newcommand*{\setOrdering}{\mathOrText{\textsc{setOrdering}}} 

\newcommand*{\EsRV}{\mathOrText{{E}^\textsc{sRV}\xspace}}
\newcommand*{\ERV}{\mathOrText{{E}^\textsc{RV}\xspace}}
\newcommand*{\ERVo}{\mathOrText{{E}^\textsc{RV}(\ospecial)\xspace}}
\newcommand*{\Ervt}{\mathOrText{{E}^\textsc{RV}(\ospecial)\xspace}} 
\newcommand*{\Ervti}[1][i]{\mathOrText{{E}^\textsc{RV}[ #1 ](\ospecial)\xspace}} 
\newcommand*{\Erspt}{\mathOrText{{E}(\rspt)\xspace}}

\newcommand*{\iCond}{\text{[RV-BC]}\xspace} 
\newcommand*{\iiCond}{\text{[RV-CC]}\xspace} 

\newcommand*{\sRVbranchingCond}{\text{[semi-RV-BC]}\xspace} 
\newcommand*{\sRVcycleCond}{\text{[semi-RV-CC]}\xspace} 

\newcommand*{\rvPutCheck}{\textsc{Constructive $\RVPUT$ Check}\xspace}

\newcommand*{\possible}{\textit{possible}}
\newcommand*{\impossible}{\textit{impossible}}
\newcommand*{\certain}{\textit{certain}}
\newcommand*{\uncertain}{\textit{uncertain}}
\newcommand*{\choosable}{\textit{choosable}}
\newcommand*{\breakable}{\textit{breakable}}

\newboolean{finishedLook}
\setboolean{finishedLook}{true}

\newlist{casesproof}{enumerate}{6}
\setlist[casesproof]{label*=\arabic*.} 

\begin{document}

    \frontmatter

\ifprintVersion
    \ifprofessionalPrint
        \newgeometry
        {
            textwidth = 134 mm,
            textheight = 220 mm,
            top = 38 mm + \extraborderlength,
            inner = 38 mm + \mybindingcorrection + \extraborderlength,
        }
    \else
        \newgeometry
        {
            textwidth = 134 mm,
            textheight = 220 mm,
            top = 38 mm,
            inner = 38 mm + \mybindingcorrection,
        }
    \fi
\else
    \newgeometry
    {
        textwidth = 134 mm,
        textheight = 220 mm,
        top = 38 mm,
        inner = 38 mm,
    }
\fi

\begin{titlepage}
    \sffamily
    \begin{center}
        \includegraphics[height = 3.2 cm]{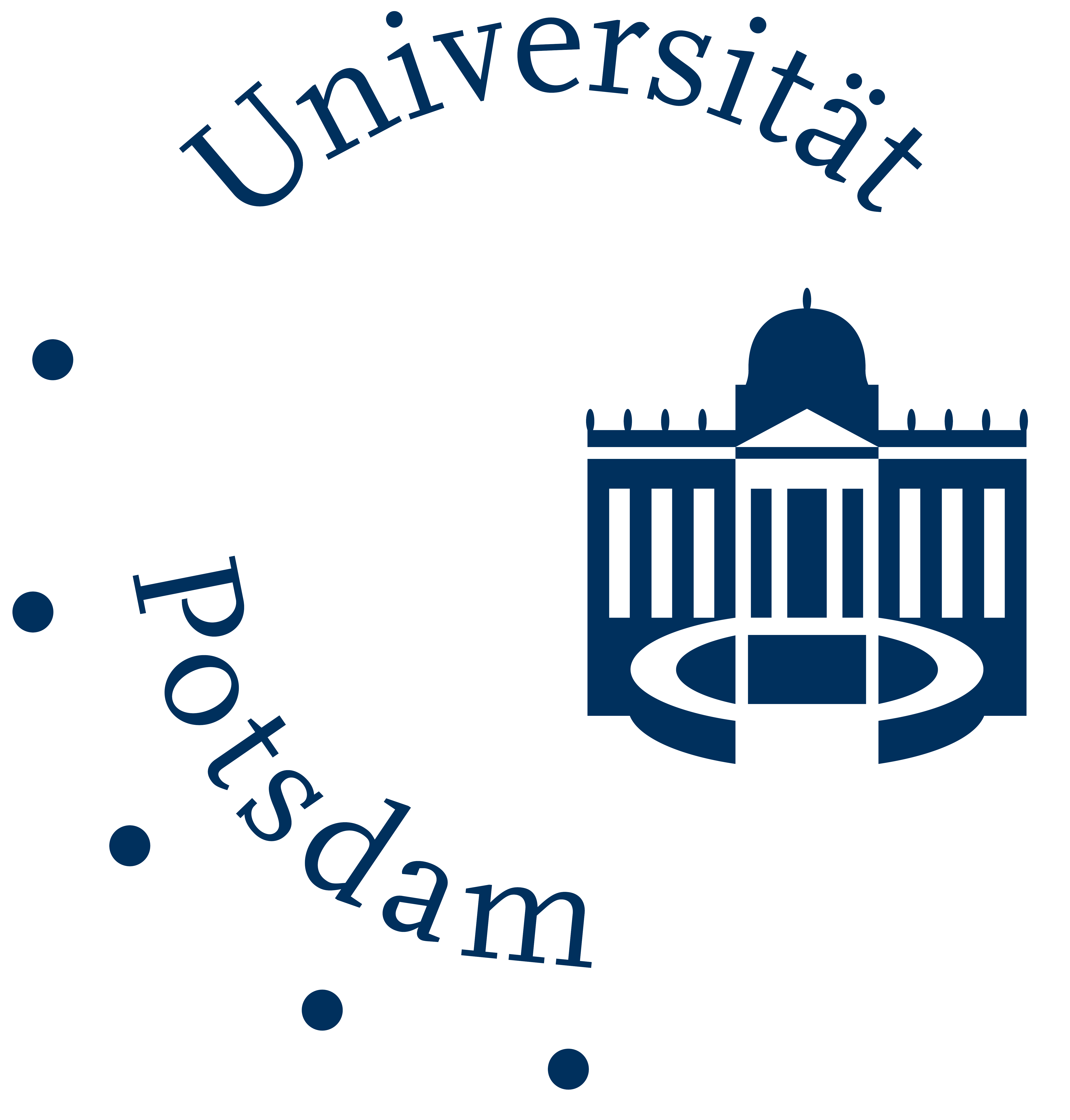} \hfill \includegraphics[height = 3 cm]{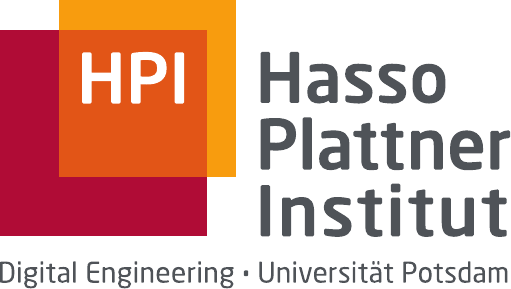}\\
        \vfil
        {\LARGE
            \rule[1 ex]{\textwidth}{1.5 pt}
            \onehalfspacing\printTitleBold\\[1 ex]
            {\vspace*{-1 ex}\Large \printGermanTitle}\\
            \rule[-1 ex]{\textwidth}{1.5 pt}
        }
        \vfil
        {\Large\textbf{\printAuthor}}
        \vfil
        {\large Universitäts\colloquialDegreeNameLowercase arbeit\\[0.25 ex]
        zur Erlangung des akademischen Grades}\\[0.25 ex]
        \bigskip
        {\Large \colloquialDegreeName{} of Science}\\[0.5 ex]
        {\large\emph{(\degreeAbbreviation\,Sc.)}}\\
        \bigskip
        {\large im Studiengang\\[0.25 ex]
        \printProgram}
        \vfil
        {\large eingereicht am \printDateReceived{} am\\[0.25 ex]
        Fachgebiet Algorithm Engineering der\\[0.25 ex]
        Digital-Engineering-Fakultät\\[0.25 ex]
        der Universität Potsdam}
    \end{center}
    
    \vfil
    \begin{table}[h]
        \centering
        \large
        \sffamily 
        {\def\arraystretch{1.2}
            \begin{tabular}{>{\bfseries}p{3.8 cm}p{5.3 cm}}
                Gutachter               & \printNameOfSupervisor\\
                Betreuende               & \printAdditionalExaminers
            \end{tabular}
        }
    \end{table}
\end{titlepage}

\restoregeometry





    \pagestyle{plain}

    \addchap{Abstract}
    Democracy relies on making collective decisions through voting. In addition, voting procedures have further applications, for example in the training of artificial intelligence. An essential criterion for determining the winner of a fair election is that all alternatives are treated equally: this is called \textit{neutrality}. The established \textit{Ranked Pairs} voting method cannot simultaneously guarantee neutrality and be computationally tractable for election with ties. \textit{River}, the recently introduced voting method, shares desirable properties with Ranked Pairs and has further advantages, such as a new property related to resistance against manipulation. Both Ranked Pairs and River use a weighted margin graph to model the election. In such a graph the vertices represent the alternatives and the edges represent how many voters prefer one alternative over another.

River can be computed as follows: iterate over the edges in descending order and insert an edge into the (initially empty) \textit{River diagram}, unless it creates a cycle or its target alternative already has an incoming edge. This computation is simple enough to do by hand.

Ties in the election can lead to edges of equal margin. To order the edges in such a case, a tiebreaking scheme must be employed. Different properties of the River voting method depend on which tiebreak is used. Many tiebreaks violate neutrality or other important properties. A tiebreaking scheme that preserves neutrality is \textit{Parallel Universe Tiebreaking} (PUT). Ranked Pairs with PUT is NP-hard to compute.

The main result of this thesis shows that River with PUT can be computed in polynomial worst-case runtime: We can check whether an alternative is a River PUT winner, by running River with a specially constructed ordering of the edges. To construct this ordering, we introduce the semi-River diagram which contains the edges that can appear in any River diagram for some arbitrary tiebreak. Next, we compute a tree that consists of strongest paths in the semi-River diagram from a selected alternative. Our specially constructed order puts edges in this tree before other equally weighted edges. We show that the selected alternative is a winner of River with the specially constructed order, if and only if it is a River PUT winner in the election.

To optimize this process, we improve the previous naïve runtime of River from $\bigO(n^4)$ to $\bigO(n^2 \cdot \log n)$, where $n$ is the number of alternatives.

    \selectlanguage{ngerman}
    \addchap{Zusammenfassung}
Die individuellen Präferenzen Einzelner mit Hilfe von Wahlen zu einer kollektiven Entscheidung zu vereinen ist die Grundlage der Demokratie. Es gibt auch weitere Anwendungsfälle hierfür, zum Beispiel beim Training von Künstlicher Intelligenz.
Ein essenzielles Kriterium für faire Wahlen ist, dass Chancengleichheit zwischen allen Alternativen besteht.
Dieses Kriterium wird als Neutralität bezeichnet.
Bei Wahlen mit Gleichständen kann die etablierte \textit{Ranked Pairs} Wahlmethode nicht gleichzeitig Neutralität gewährleisten und effizient berechenbar sein.
Die kürzlich entwickelte \textit{River} Wahlmethode hat verbesserte Resistenz gegen Wahlmanipulation und erfüllt ähnliche wünschenswerte Eigenschaften wie Ranked Pairs.
Sowohl River als auch Ranked Pairs nutzen einen Margin Graphen, welcher die Alternativen durch Knoten repräsentiert und eine gewichtete gerichtete Kante von einer Alternative zur anderen aussagt, wie viele Wählende die eine Alternative über die andere präferieren.
River wird berechnet, indem über die Kanten in absteigender Reihenfolge des Gewichts iteriert wird, und eine Kante nur in das anfangs leere River Diagramm eingefügt wird, falls es mit dieser Kante noch azyklisch ist und jeder Knoten maximal eine eingehende Kante hat.
Dieser Prozess ist einfach genug um ihn auch ohne Computer zu berechnen.
Kanten im Margin Graphen können das gleiche Gewicht haben, was Gleichstände in der repräsentieren Wahl darstellt. In diesem Fall muss ein Tiebreak eingesetzt werden, um diese Kanten zu sortieren. Mehrere Eigenschaften  von River hängen von den Eigenschaften dieses Tiebreaks ab und viele Tiebreaks verletzen die Neutralität der Wahl. Der Parallel Universe Tiebreak erhält allerdings die Neutralität.

Das Hauptergebnis dieser Arbeit ist, dass River mit dem Parallel Universe Tiebreak (River PUT) mit polynomieller worst-case Laufzeit berechenbar ist. Dies ist anders als bei Ranked Pairs.
Wir erreichen diese Laufzeit durch das Ausführen von River für jede Alternative mit einer speziellen Sortierung der Kanten.
Um diese Sortierung zu konstruieren, erstellen wir zunächst das semi-River Diagramm welches Kanten enthält die in einem River Diagramm mit einem beliebigen Tiebreak sind. Dann berechnen wir einen Teilbaum von diesem Graphen, der die stärksten Pfade von einer gewählten Alternative enthält. Unser speziell konstruierter Tiebreak sortiert Kanten, welche in diesem Baum sind vor anderen mit gleichem Gewicht.
Wir beweisen, dass die gewählte Alternative von River mit dieser speziellen Sortierung der Kanten als beste Alternative ausgewählt wird, genau dann, wenn sie auch von River PUT ausgewählt wird. 

Um einen Teil dieses Prozesses zu optimieren, stellen wir einen Algorithmus vor der River in $\bigO{n^2 \cdot \log n}$ mit der Anzahl der Alternativen $n$ berechnet, was eine Verbesserung gegenüber der naiven Laufzeit von River in $\bigO{n^4}$ darstellt.
    \selectlanguage{american}

    \addchap{Acknowledgments}
First and foremost I have to thank my advisors, Michelle and Stefan. Reading my final version it is very clear how significantly your feedback and advice improved this thesis and my abilities.
I would also like to thank my office mates during the bachelor's project and during the three final weeks for reminding me to eat and take breaks.
Of course, my family and friends enabled me to work on this thesis as much as I did by providing distractions through long adventures and short volleyball matches. 

    \renewcommand*\contentsname{Table Of Contents}
    \setuptoc{toc}{totoc}
    \tableofcontents

    \pagestyle{headings}
    \mainmatter
    \chapter{Introduction}


\section*{Motivation}
Democracy is the ''government by the people'', where citizens have equal rights and collectively make decisions \cite{oxfordenglishdictionaryDemocracy2023}.
Luckily, \SI{45.4}{\percent} of the world's population lives in a democracy. 
However, even though almost \SI{25}{\percent} of the world population is expected to vote in 2024, only roughly \SI{8}{\percent} of the global population live in a full democracy, based on truly free and fair elections \cite[Table 1]{eiulimitedDemocracyIndex2023} \cite{buchholzInfographic2024Super2024}. 
Elections are clearly integral to fully functioning democracies. Thus, the \textit{voting methods} that are used to determine the outcome of those elections need to be closely examined.
Doing this in a scientific way is one aspect of Social Choice Theory.
The field studies how we can ''aggregat[e] individual preferences toward a collective choice'' using multidisciplinary perspectives from mathematics, economics, political science, and most recently computer science  \cite[Ch. 1]{brandtHandbookComputationalSocial2016}.
How a winner of an election is determined can be modeled as a mathematical function, where the preferences of the voters over the alternatives are given and the winner of the election is selected among the alternatives. This is called a social choice function. We will use the terms social choice function and voting method interchangeably.
Which voting method is used, has a crucial impact on the fairness of the election and its effectiveness in selecting the best alternative.

Aggregating preferences toward a collective choice is also important beyond political elections. One example is the curation of outputs of large language models to train artificial intelligence \cite{10.5555/3666122.3668186}. 

The most widely used voting methods in politics are based on Plurality voting. The democratic system in Germany uses Plurality voting \cite{brandtComputationalSocialChoice2023}. In 1997 Plurality-majority electoral systems represented \SI{59}{\percent} of the global voting population (2.44 billion people) \cite{theelectoralknowledgenetworkGlobalDistributionElectoral2005}.
In Plurality voting a voter can give one vote to one of the candidates. The candidate with the highest number of votes wins. This has many errors.

\citeauthor{hamlinTopWaysPlurality2015}, an organization that promotes a different form of voting, argues that having to pick one top choice is very restricting on the voter's expression of their opinions.
It is prone to the ''Spoiler Effect'', where a candidate with very few votes, might take votes away from a vastly more popular candidate, who would have otherwise had an absolute majority.
Voters who want to avoid this effect, might misrepresent their opinion and not vote for their favorite candidate. 
This could also make it more difficult for new candidates to gain votes. 
Plurality voting may also disadvantage similar candidates, leading to the ''Centre-Squeeze Effect'', where moderate votes are split between many moderate candidates, while extremist candidates appear to get more support. \cite{hamlinTopWaysPlurality2015}

Social Choice experts often disagree about the best voting method. However, they do agree that Plurality voting is not a good choice for democratic elections. This is also shown by a small poll among 22 experts in 2010, where no one approved Plurality voting \cite{laslierLoserPluralityVoting2012}.

\begin{figure}[p]
    \centering
    \includegraphics[width=\linewidth]{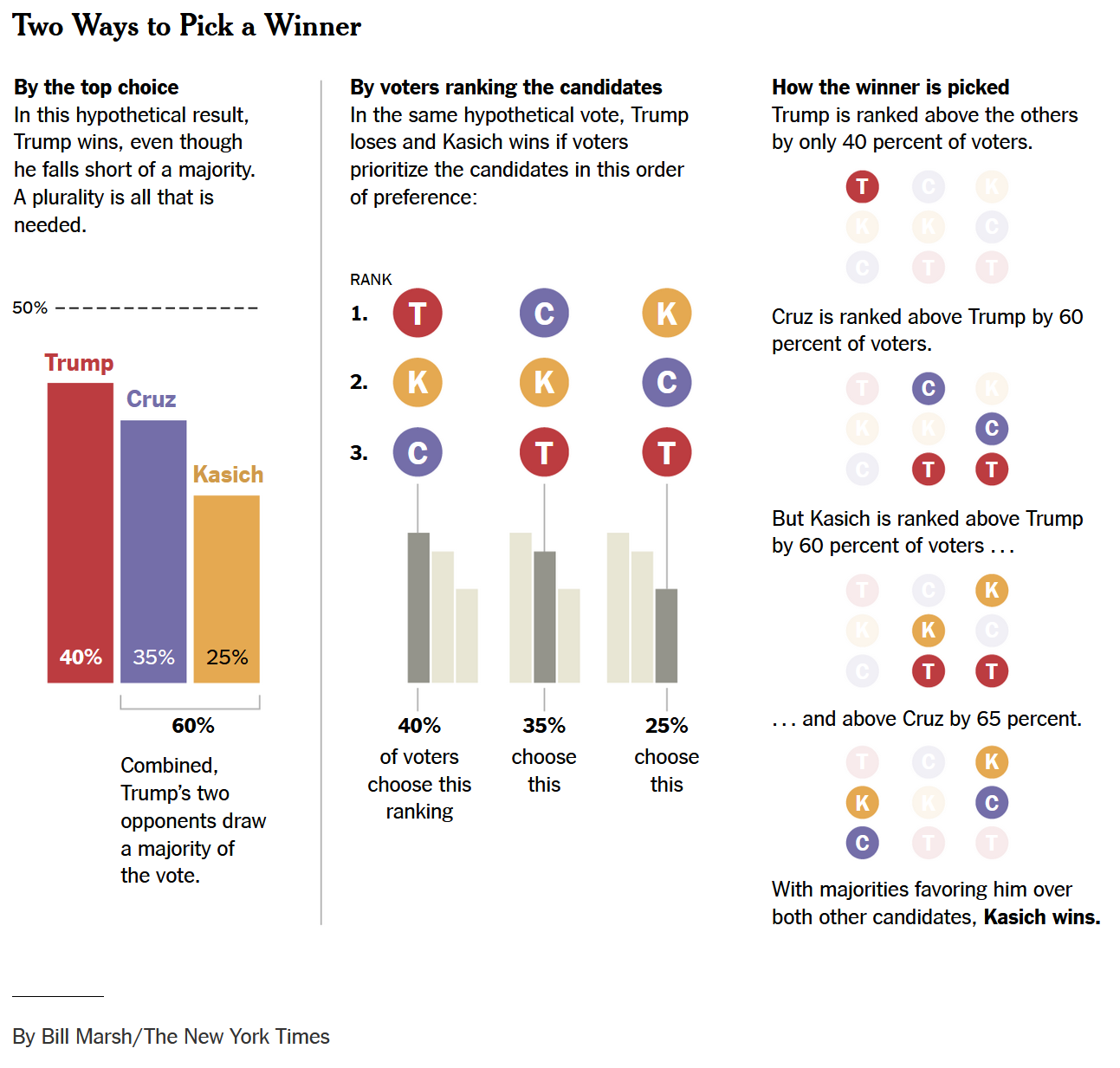}
    \caption{A real-world example where the Plurality voting result might have been against the preferences of a majority of voters, according to an opinion article by the economics Nobel laureates \textcite{maskinOpinionHowMajority2016a}. It also illustrates the difference between plurality voting (left) and ranked voting (right).}
    \label{fig:nytTrumpExamplePlurality}
\end{figure}
One of many real-world instances where Plurality voting may have produced a result that a majority of people found undesirable, is the 2016 Republican Primary Election. Donald Trump won, even though John Kasich may have been the choice that satisfied the most voters, while Donald Trump may have been the least favorite candidate of a majority of people \cite{maskinOpinionHowMajority2016a}. This is illustrated in \Cref{fig:nytTrumpExamplePlurality}. The figure also shows why a different voting system may have been better able to represent the collective opinion.

\section*{A Chronological Introduction to Social Choice}
Early discussions of democracy and elections already occurred in ancient Rome. The initial scientific approach occurred centuries later, during the Enlightenment.
In the 20th century came a more mathematical approach, based on axiomatic reasoning.
Through the later combination with computer science, the field of computational social choice came into existence. \cite{brandtHandbookComputationalSocial2016}  


\subsection*{Core Ideas of Social Choice Theory}
Two important figures in the early era of Social Choice Theory are Nicolas Condorcet (1743-1794) and Jean-Charles de Borda (1733-1799). While de Borda had the initial ideas for \textit{positional scoring rules}, Condorcet argued for \textit{pairwise majority comparisons}.


Both concepts require more information about the voters preferences than we may be familiar with from Plurality voting. In an election, a voter has preferences among the available alternatives. Thus, the ballot of a voter we consider in this work does not only contain information about their favorite alternative but clarifies their exact preference order among all alternatives. An illustration of the information this gives us for an election - or as we will later call it: The preference profile - can be seen in \Cref{tab:examplePreferenceProfile}. In \Cref{tab:examplePreferenceProfile} the first column shows, that 4 of 11 voters think Alice is the best alternative, Bob the second best, and Charlie the third best.
\begin{figure}[p]
\centering
\begin{subfigure}[b]{0.475\textwidth}
        \centering
        \begin{tabular}{@{}cccc@{}}
        \toprule
        \multicolumn{1}{c}{4} & \multicolumn{1}{c}{3} & \multicolumn{1}{c}{2} & \multicolumn{1}{c}{2} \\ \midrule
        Alice                 & Bob                  & Bob                  & Charlie                 \\
        Bob                  & Charlie                 & Alice                 & Alice                 \\
        Charlie                 & Alice                 & Charlie                 & Bob                  \\ \bottomrule
    \end{tabular}
    \caption{\small 3 alternatives; 11 voters \cite[Ch. 1]{brandtHandbookComputationalSocial2016}}
    \label{tab:examplePreferenceProfile}
\end{subfigure}
\hfill
\begin{subfigure}[b]{0.475\textwidth}  
    \centering 
    \begin{tabular}{@{}ccccc@{}}
        \toprule
        \multicolumn{1}{c}{4} & \multicolumn{1}{c}{3} & \multicolumn{1}{c}{2} & \multicolumn{1}{c}{4} \\ \midrule
        Alice                 & Bob                  & Bob                  & Charlie                 \\
        Bob                  & Charlie                 & Alice                 & Alice                 \\
        Charlie                 & Alice                 & Charlie                 & Bob                  \\ \bottomrule
    \end{tabular}
    \caption{\small 3 alternatives; 13 voters \cite[Ch. 1]{brandtHandbookComputationalSocial2016}}
    \label{tab:examplePreferenceProfileMod}\end{subfigure}
\vskip\baselineskip
\begin{subfigure}[b]{0.475\textwidth}   
        \centering 
    \includegraphics[scale=2]{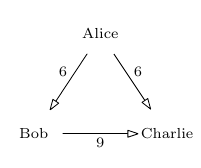}
    
    \caption{\small The margin graph of the election from \Cref{tab:examplePreferenceProfile}. The edges are the result of the pairwise majority comparison between the two alternatives they connect. Here Alice is the Condorcet winner.}
    \label{fig:exampleMarginGraph}
\end{subfigure}
\hfill
\begin{subfigure}[b]{0.475\textwidth}   
    \centering 
    \includegraphics[scale=2]{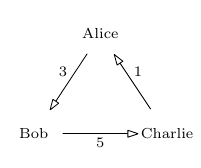}
    
    \caption{\small The margin graph of the modified election from \Cref{tab:examplePreferenceProfileMod}. The edges are the result of the pairwise majority comparison between the two alternatives they connect. Here there is no Condorcet winner.}
    \label{fig:exampleMarginGraphMod}
\end{subfigure}
\vskip\baselineskip
\begin{subfigure}[b]{0.475\textwidth}   
        \centering 
    \begin{tabular}{ll}
            \textbf{Plurality}  & Bob with\\
            & $0\cdot4+1\cdot3+1\cdot2+0\cdot2$\\
            & $=5$ Points\\
            \textbf{Borda}      & Bob with\\
            & $2\cdot4+3\cdot3+3\cdot2+1\cdot2$\\
            & $=25$ Points
        \end{tabular}   
    \caption{\small The positional scoring rule results of \Cref{tab:examplePreferenceProfile}.}
    \label{fig:daBordaAndPlurality}
\end{subfigure}
\hfill
\begin{subfigure}[b]{0.475\textwidth}   
        \centering 
    \begin{tabular}{ll}
            \textbf{Plurality}  & Bob with\\
            & $0\cdot4+1\cdot3+1\cdot2+0\cdot4$\\
            & $=5$ Points\\
            \textbf{Borda}      & Bob and Alice with\\
            & $2\cdot4+3\cdot3+3\cdot2+1\cdot4$\\
            & $3\cdot4+1\cdot3+2\cdot2+2\cdot4$\\
            & $=27$ Points
        \end{tabular}   
    \caption{\small The positional scoring rule results of \Cref{tab:examplePreferenceProfileMod}.}
    \label{fig:daBordaAndPluralityMod}
\end{subfigure}
\caption[ Examples of Elections ]
        {\small Two examples of elections. One with an Condorcet winner (left) and one to illustrate the Condorcet paradox (right) } 
        \label{fig:introExamplesCondorcet}
\end{figure}
Positional scoring rules use this order to award points to an alternative based on its placement in the preference for every individual voter. In an election with $m$ alternatives for each voter, the \textit{Borda rule} gives $m$ points to the alternative ranked first, $m-1$ points to the alternative ranked second, and so on. The alternative with the most points summed over all voters wins the election. 
Plurality voting can be seen as a positional scoring rule: In Plurality voting, only the favorite alternative of each voter gets $1$ point, and the alternative with the most points wins.

{Condorcet critiqued the Borda rule and Plurality voting. One example of why is the election in \Cref{tab:examplePreferenceProfile}. Bob is the winner of this election under the Borda rule and with Plurality voting in \Cref{fig:daBordaAndPlurality}. However, a majority of 6 voters (more than half of the voters) prefer Alice to Bob.}
This is why Condorcet proposed to extract how an election between only two of the alternatives would have gone from the ranked ballots. This comparison is illustrated by the margin graph in \Cref{fig:exampleMarginGraph}. In a \textit{margin graph}, the vertices represent the alternatives and the directed edges represent one alternative beating another alternative in the pairwise majority comparison with a certain margin. By calculating these pairwise majorities for every pair of alternatives the winner using this method is sometimes very hard to dispute:
In \Cref{tab:examplePreferenceProfile} 6 voters also prefer Alice to Charlie and thus Alice wins all pairwise majority comparisons. Alice is called the \textit{Condorcet winner} of this election. In \Cref{fig:exampleMarginGraph} we can see that Alice is the Condorcet winner because they only have outgoing arrows.
The \Cref{fig:nytTrumpExamplePlurality} also illustrates this concept for the previous real-world example.

However, sometimes the collective preference of the voters in an election might be ''irrational''. With two additional voters preferring Charlie over Alice over Bob added to the preference profile in \Cref{tab:examplePreferenceProfile}, there no longer is a Condorcet winner. This modified election can be seen in \Cref{tab:examplePreferenceProfileMod}.
In the margin graph \Cref{fig:exampleMarginGraphMod} there is this contradictory cycle, where the voters collectively prefer Alice over Bob and Bob over Charlie, but also Charlie over Alice. This is called the \textit{Condorcet paradox}, where there is no Condorcet winner. There are many different voting methods proposed to resolve this. Voting methods, which result in the Condorcet winner if there is one, are called \textit{Condorcet consistent}. 

\subsection*{The Axiomatic Approach}
Well after these early ideas came the classical period of Social Choice Theory in the 1950s. This era is characterized by a rigorous mathematical approach. Most predominantly put forward by Kenneth Arrow who developed a formal method to ''speak about and analyz[e] all possible [voting] rules'' \cite[Ch. 1]{brandtHandbookComputationalSocial2016}. 
This axiomatic approach in Social Choice Theory uses mathematical formulations of rules or properties that a voting method should satisfy (referred to as axioms). The axioms are the result of philosophical 
discussions and can be used as mathematical statements. From these the logical conclusions can be explored, using formal mathematical reasoning \cite[Ch. 1]{brandtHandbookComputationalSocial2016}.
This approach allows for better comparison of different voting rules, by examining which properties or axioms they satisfy.
One axiom we already discussed is Condorcet consistency, where a social choice function always selects the Condorcet winner if there is one.
The ''Spoiler Effect'' and the ''Center-Squeeze Effect'', which Plurality voting has, are also modeled by properties to prevent them in Social Choice Theory: They are intuitively comparable to the Independence of Irrelevant Alternatives (IIA) and the Independence of Clones (IoC) respectively. For some voting methods, it can be mathematically proven, that they are resistant to these negative effects.
Further basic but essential properties are \textit{anonymity} and \textit{neutrality}, where all voters and alternatives respectively are treated equally.
Another example of a desirable property of social choice functions is \textit{monotonicity}, where if some voter's opinion of a winning alternative improves and it is put in a higher position on that voter's ballot, then this alternative is still a winner \cite{doringFlowingFairnessRiver2024}.

However, Arrows research also shows that finding the ultimate voting method is not possible. 
In his famous theorem, he shows that no ranked voting method can satisfy \textit{positive association} (which has a similar intuition to monotonicity), IIA and \textit{non-dictatorship} or \textit{non-imposition} at the same time \cite[Possibility Theorem]{arrowDifficultyConceptSocial1950}. This means that choosing a voting method is often a trade-off between different axioms. Thus, finding methods with better trade-offs is still an active research field. 



\subsection*{Computational Social Choice}
With the rise of computers came the ''Computational Turn''  \cite[Ch. 1]{brandtHandbookComputationalSocial2016}. Most previous research ''neglected the computational effort ''  \cite[Ch. 1]{brandtHandbookComputationalSocial2016} needed to compute the results of elections with a proposed voting method. However, this is essential to implementing a voting method for real-world applications. A method taking years to compute a winner cannot find practical application in democratic elections.
 \cite[Ch. 1]{brandtHandbookComputationalSocial2016}

This work can be attributed to the field of Computational Social Choice Theory, as we examine the computational complexity of a novel voting method.
For a more extensive history of the field, refer to \cite{brandtHandbookComputationalSocial2016} or \cite{listSocialChoiceTheory2022a}.
\section*{Computational Complexity}
Here we will give a brief introduction to the Complexity Theory relevant for this work. For a more extensive explanation, refer to \cite[Ch. 7]{sipserIntroductionTheoryComputation2013}.
A solution for a \textit{decision problem} in the computational sense is, whether or not some element is part of a set with certain properties. How hard a decision problem is, depends on the defined properties of this set. An example of a decision problem in this thesis is, whether an alternative in an election, given the preferences of the voter, is in the set of winning alternatives of a social choice function.
For our purposes, a \textit{complexity class} is a set of decision problems, where deciding the answer for a given input has at most a certain computational complexity.

The measure of computational complexity relevant for this work is the \textit{worst-case runtime complexity}. This means that the number of operations an algorithm does to compute the solution for a problem with any input is bounded by some function. This is usually expressed by saying that an algorithm's runtime is in $\bigO{f}$ for some function $f$ of the input. If this function grows quickly, executing the algorithm may take a long time, making it unusable for practical purposes.

The most famous complexity classes are \textit{P} and \textit{NP}. The problems in P are solvable with a runtime polynomial in the size of the input, using only deterministic computation. This means for a decisions problem in P, there exists an algorithm for a ''normal'' computer, which can decide if an input is a member of the set of the problem or not, with a worst-case runtime that is a polynomial function of the input size. The problems in NP are solvable with a polynomial runtime using non-deterministic computation. Since no way is known to do non-deterministic computation efficiently on ''normal'' computers, the only deterministic algorithms known to decide these problems have a worst-case runtime which is exponential in the input size. 

P is a subset of NP. It is not known if $\text{P}=\text{NP}$. Problems that are known to be in P seem easier to compute than problems that are so far only known to be in NP. However, there is no problem which is known to be in NP and known not to be in P. Also, no way has been found to solve all problems in NP fast enough so that they are also in P.

A problem is called \textit{hard} for a complexity class, if it is at least as hard to compute a result of that problem, as for all problems in this complexity class.
This means if a problem is shown to be NP-hard, it is at least as hard as all problems in NP. 
Take as an example a voting method, where determining the set of winners is NP-hard. For that voting method, only algorithms have been found, which may take a very long time to compute the result on a sufficiently large input.
 
In this work, we show the membership in P of deciding the winners of a certain voting method. We do this by giving an algorithm that can be used to decide the winners of this voting method and show it has a runtime bounded by a polynomial of the algorithm's input size.
\section*{New Voting Methods}
Modern Social Choice scientists also propose new voting methods, which better satisfy different axioms. With the new perspective of computer science and contributions from economists, mathematicians, political scientists, and philosophers, the field is gaining popularity.  \cite[Ch. 1]{brandtHandbookComputationalSocial2016}



However, one key issue in coming up with new voting methods is the already mentioned Condorcet paradox, where there may not be a Condorcet winner in every election. In those cases, it is not obvious what a voting method should do. \textit{Condorcet consistent} voting methods pick the Condorcet winner if it exists and otherwise try to find the best candidate by some other metric, which may be inferred from the votes cast or other additional information \cite[Ch. 1]{brandtHandbookComputationalSocial2016}.
One way to find a winner in cases without a Condorcet winner uses the information provided by the margins of the edges in the margin graph: Methods which do this are called \textit{margin based methods}. This thesis examines such a method.

Before we introduce this method, we need the concept of \textit{immunity}. Immunity can be seen as a generalization of Condorcet winners. An alternative with immunity against majority complaints has, for each majority defeat by another alternative, a path in the margin graph to this alternative, which is at least as strong as the margin of its defeat.
In \Cref{fig:exampleMarginGraphMod} Alice has immunity, because, while they are beaten by Charlie with a margin of 1, there is a path from Alice to Charlie, which is stronger. Note that Condorcet winners also have immunity, since they have no incoming edges. The advantage of immunity is that, unlike the Condorcet winner, at least one immune alternative always exits \cite{doringFlowingFairnessRiver2024}. 

\subsection*{Split Cycle and Refinements}
A voting method that relies on immunity is \textit{Split Cycle} by \textcite{hollidaySplitCycleNew2023}. The set of winners of the Split Cycle voting method for an election is the set of immune alternatives in the margin graph. Split Cycle works by removing the edges with the lowest margin from every cycle in the margin graph. An alternative with no incoming edges must then have immunity and is a winner.
There is always at least one immune alternative, unlike with the Condorcet winner. \cite{hollidaySplitCycleNew2023}

\textit{Refinements} of Split Cycle are voting methods, where the winners are always some subset of the winners of Split Cycle. These may be advantageous in cases where the refinements are more resolute than Split Cycle. That means they select only a single winner instead of multiple winners. They may also satisfy desirable axioms that Split Cycle does not satisfy. 

One of these is \textit{Ranked Pairs}. It works by inserting the edges of the margin graph ordered by decreasing margin into an initially empty graph and omitting those edges, which would create a cycle in the new graph. Any alternative without an incoming edge in that graph is a Ranked Pairs winner. One motivation for Ranked Pairs is that, unlike Split Cycle, it satisfies the Independence of Clones (IoC) criterion, where nearly identical alternatives should not hinder each other in the election. \cite{doringFlowingFairnessRiver2024}

A further criterion is the Independence of Pareto Dominated Alternatives (IPDA). An alternative is Pareto dominated by another if all voters place it lower than the other. The IPDA criterion states: If such an alternative is removed from the election, that should not change the outcome. 
Surprisingly Split Cycle and Ranked Pairs as well as other Split Cycle refinements violate IPDA. \cite{doringFlowingFairnessRiver2024}  

\subsection*{River and the Goal of this Thesis}
This thesis is mainly concerned with \textit{River}: a novel Split Cycle refinement. River is very similar to Ranked Pairs.
It satisfies Condorcet consistency and monotonicity. River is resolute for elections without ties and satisfies IoC. It also satisfies the other axioms satisfied by Split Cycle. 
An essential improvement is that it further satisfies IPDA.
A common issue with sophisticated voting methods is that they may be unsuited for democratic elections because it is hard for people to understand why a winner is selected. River provides a subtree of the margin graph as a certificate, which proves the immunity of the selected winner, making it easy to justify why the selected winner is the best alternative. \cite{doringFlowingFairnessRiver2024}

However, both anonymity and neutrality of River depend on how ties are resolved, in the same way as for Ranked Pairs \cite{doringFlowingFairnessRiver2024}.
In real-world elections ties in the margin graph can occur (a harmless example of this is \Cref{fig:exampleMarginGraph}). When there are edges with the same margin a tiebreak is needed, since Ranked Pairs and River process the edges in order of decreasing margin.
So, while River and Ranked Pairs are anonymous and neutral for elections without ties, for general preference profiles, each fixed tiebreak which selects a single alternative violates anonymity or neutrality \cite{doringFlowingFairnessRiver2024,brillPriceNeutralityRanked2021}. Other properties of River also depend on a suitable tiebreak \cite{doringFlowingFairnessRiver2024}.
One generalized scheme for tiebreaking is Parallel Universe Tiebreaking (PUT). In PUT all possible ways to break the ties are examined, and if an alternative wins with some tiebreak, it is a winner under PUT \cite{brillPriceNeutralityRanked2021}. A more generalized definition of PUT can be found in \cite{freemanGeneralTiebreakingSchemes2015}, however, the definition used by \textcite{brillPriceNeutralityRanked2021} will suffice for this work.
Parallel Universe Tiebreaking preserves neutrality for Ranked Pairs \cite{brillPriceNeutralityRanked2021}.
Ranked Pairs is NP-hard with Parallel Universe Tiebreaking, thus there is a tradeoff between computational tractability and neutrality for Ranked Pairs \cite{brillPriceNeutralityRanked2021}.

The main result of this thesis is, that River with Parallel Universe Tiebreaking is in P. The algorithm we give in this thesis can be used to compute the set of PUT winners of River and has polynomial runtime.

    \chapter{Preliminaries}

    \section{Social Choice Elections Basic Definitions}
To model a democratic election, we consider $m\geq1$ voters with preferences over a set \alts of $n\geq2$ alternatives. For $n\in\mathbb{N}$ we write $[n]$ instead of $\sset{1, \dots, n}$. The preference of a voter $i \in [m]$ is represented by a linear order $\succ_i$ over \alts, where the voter ranks their favorite alternative first, second favorite next, and so on.
A preference profile $\prefs=(\succ_1, \dots, \succ_m)$ is the collection of all the preferences of $m$ voters.
When referring to the set of all possible preference profiles, we use \allPrefs. 
We define the \textit{majority margin} for $\prefs \in \allPrefs$ of alternative $x \in \alts$ over alternative $y \in \alts$ as $$\mpref{x,y}=\abs{\msset{1\leq i\leq m \mid x \succ_i y}} - \abs{\msset{1\leq i \leq m \mid y \succ_i x}}$$ to compare the collective preference between different candidates. This majority margin represents, how often one alternative is ranked above another in all preferences in \prefs.
The weighted \textit{margin graph} of a preference profile \prefs is defined as $$\mG(\prefs) = (\alts, \sset{(x,y) \mid \mpref{x,y} \geq 0}).$$
Note that for ties between alternatives, there are bidirectional edges with a margin of zero.
When the preference profile \prefs is clear from context, we write \mG as shorthand for $\mG(\prefs)$ and $m$ instead of $m_{\prefs}$.
We denote the edges of the margin graph by $E(\mG)$. We use this notation analogously for other graphs. 

Given a preference profile \prefs over alternatives \alts, a \textit{social choice function} selects a non-empty set of winners $\alts' \subseteq \alts$.
\subsection*{Paths in the Margin Graph}
To define the social choice function River, we need the notion of a \textit{majority path} from $x_1$ to $x_l$ in a margin graph $\mG(\prefs)$ with $\prefs \in \allPrefs$, which is a sequence $\Path{x_1}{x_l}=\langle x_1, \dots, x_l\rangle$ of distinct alternatives such that for $i\in \sset{1, \dots, l-1}$, $x_i$ directly defeats $x_{i+1}$ (that is $(x_i, x_{i+1}) \in E(\mG)$ or equivalently $\m{x_i, x_{i+1}} \geq 0$).
The \textit{strength} of such a path $\Path{x_1}{x_l}$ is the smallest margin of its edges: 
$$\str(\Path{x_1}{x_l}) = \min_{1 \leq i \leq l-1}(\m{x_i, x_{i+1}}).$$
We call an edge $e$ on a path $\PathP$ with $\m{e} = \str(\PathP)$ a minimum edge of that path. 
We denote a subpath of \Path{x_1}{x_l} from $x_i$ to $x_j$ by $\Path{x_1}{x_l}[i:j]$. 
We write $\Path{1}{x_i}\circ \Path{x_i}{x_l}= \Path{1}{x_l}$ for the concatenation of the paths $\Path{1}{x_i}$ and $\Path{x_i}{x_l}$.

\subsection*{Preliminary Graph Definitions}
For a graph $G=(V,E)$ given an edge $e=(x,y)\in E$, a \textit{cycle} goes through $e$ if there is a path from $y$ to $x$ in $G$. An \textit{isolated} vertex has no incoming and no outgoing edges.

We define a rooted \textit{tree} as a graph $T=(V,E)$, where there is a unique path from the root $s \in V$ to any $ v\in V$.
Note that a tree contains no cycles and the in-degree of every $v\in V$ is at most one. The root only has outgoing edges.
A \textit{forest} is a graph $F=(V,E)$ containing one or more trees. Note that an isolated vertex is also a tree.
A \textit{directed acyclic graph} (DAG) is a directed graph $G=(V,E)$ where there are no cycles. Formally, for all $(u,v) \in E$ there is no path from $v$ to $u$ in $G$. A \textit{source} of a DAG is any vertex with no incoming edges.

A \textit{breadth first search} from $u \in V$ denoted by $\text{BFS}(u)$ on the directed graph $G=(V,E)$, yields the set of all vertices which are reachable from $u$ via some path in $G$. BFS has a runtime in $\bigO{\abs{V} + \abs{E}}$ \cite[Ch. 22.2]{cormenIntroductionAlgorithms2022}.
The \textit{transposed graph} $G^{-1}=(V,E')$ for a directed graph $G=(V,E)$, is the graph where for all $(u,v) \in E$ the edge $(v,u)\in E'$.
The \textit{ancestors} of $u\in V$ are the vertices in $\text{BFS}(u)$ on $G^{-1}$ for $G=(V,E)$. This means they have some path to $u$ in $G$.
The ancestors of $u\in V$ up to $v \in V$ are the vertices in a $\text{BFS}(u)$ on $G^{-1}$ for $G=(V,E)$ that does not take any outgoing edges of $v$ in $G^{-1}$. So, this is the set of vertices that have paths to $u$ in $G$ that do not use an incoming edge of $v$.

\subsection*{Tiebreaking in Social Choice Functions}
A \textit{descending linear ordering} is an ordering $\lino = (e_1, \dots, e_{\abs{E}})$, such that for all $1 \leq i < j \leq \abs{E}: \m{e_i} \geq \m{e_j}$.
For $o_i$ in $o=(o_1, \dots, o_l)$ we write $o[i]$. We also write $\text{inv}_o(o_i)$ for $i$. 
The set of all descending linear orderings of $E(\mG)$ is $\allDescOrders$. If there is more than one descending linear ordering possible ($\abs{\allDescOrders}>1$) due to ties in the preference profile, one such ordering $\lino \in \allDescOrders$ is a tiebreak. 

The \textit{Parallel Universe Tiebreaking} (PUT) result of a social choice function is the set of alternatives, where an alternative is a winner under some tiebreak.
So the Parallel Universe Tiebreaking result in our case is the union over the winners of our social choice function under all descending linear orderings $\lino \in \allDescOrders$. 



\section{Definition of River}
\textcite{doringFlowingFairnessRiver2024} introduce the \textit{River} social choice function $\RV\colon \allDescOrders \times \allPrefs \to \mathcal{P}(\alts)$.
River takes the margin graph \mG of a preference profile \prefs and an $\lino \in \allDescOrders$ over the edges $E(\mG)$ as an input. It outputs the candidate with no incoming edges in the River diagram $\mRV(\lino)$.
We denote by $\mRVi[i]$ the River diagram after processing the edge $\lino[i]$ and write $\mRV(\lino)$ for the River diagram after this procedure instead of $\mRVi[\abs{E(\mG)}](\lino)$.
In this work, we examine the River diagram.
\begin{definition} \label{def:river}
    \textit{River} creates the River diagram as follows:
    \begin{enumerate}
        \item Initialize the River diagram $\mRVi[0]$ as a directed graph $(\alts, \emptyset)$.
        \item For $i$ from $1$ to $\abs{E(\mG)}$, process the edge $e =(x,y) := \lino[i]$ and add $e$ to \mRVi[i], unless the \textbf{River branching condition \iCond } or the \textbf{River cycle condition \iiCond } is satisfied:
        \item[]
            \begin{tabular}{rl}
            \textbf{\iCond}     & There already is an incoming edge for $y$ in \mRVi[i-1]. \\ 
            \textbf{\iiCond}    & There is a path from $y$ to $x$ in \mRVi[i-1]. 
            \end{tabular}
    \end{enumerate}
\end{definition} 
From the processing of edges in order of decreasing margin by $\lino$ in River follows:
\begin{corollary}\label{obs:atLeastStrengthEdgesToBreakEdge}
    When the edge $\lino[i]=(x,y)$ in $E(\mG)$ is processed by River, for \iCond to be satisfied, the other incoming edge to $y$ must have a higher or equal margin than $\lino[i]$, and for \iiCond to be satisfied, there must be a path from $y$ to $x$ of higher or equal strength in \mRV.
\end{corollary}
Further, the following property of \mRV was proven by \textcite{doringFlowingFairnessRiver2024} for uniquely weighted preference profiles and for general profiles with a tiebreaker. 
\begin{corollary}\label{obs:propertiesOfRiverDiagram} 
    For $\prefs \in \allPrefs$ with margin graph \mG, and any $\lino \in \allDescOrders$, the River diagram $\mRV(\lino)$ is a rooted tree and the unique root of $\mRV(\lino)$ is the River winner.
\end{corollary}
In later proofs, we argue about the intermediate River diagram during its construction. Note that an edge is only added to the initially empty River diagram if it does not create a cycle and every alternative has at most one incoming edge and thus:
\begin{corollary}\label{river:intermediateForest}
    For $\prefs \in \allPrefs$ with margin graph \mG and $\lino \in \allDescOrders$ for all $0 \leq i \leq \abs{E(\mG)}$ the intermediate River diagram $\mRVi[i](\lino)$ is a forest. This means it does not contain cycles and every alternative has at most one incoming edge.
\end{corollary}


In this thesis, we examine River with Parallel Universe Tiebreaking.
\begin{definition}
We define River with Parallel Universe Tiebreaking (PUT) as $\RVPUT\colon \allPrefs \to \mathcal{P}(\alts)$ with $\RVPUT(\prefs) = \bigcup_{\lino \in \allDescOrders } \RV(\lino, \prefs)$.
\end{definition}
How to compute this efficiently is not obvious, since the size of \allDescOrders is the number of permutations of sets of edges with the same margin. We show this is possible.



An important property of River winners is \textit{immunity}. For every other alternative that beats them in a pairwise majority comparison, there exists a majority path as least as strong as this pairwise majority strength \cite{doringFlowingFairnessRiver2024}. From this follows:
\begin{corollary}
 \label{lma:PUTwinnersBreakAllIncomingEdges}
 Any $\RVPUT(\prefs)$ winner \wcand with incoming edge $e_{\text{defeat}}=(x,\wcand)$ in $E(\mG)$ has a path $\Path{\wcand}{x}$ from $\wcand$ to $x$ with $\str(\Path{\wcand}{x}) \geq \m{e_{\text{defeat}}}$ in every River diagram where they are the winner, for $\prefs \in \allPrefs$.
\end{corollary}
Thus, observe that for a $\RVPUT$ winner, there is a River diagram, where all incoming edges they have in the margin graph are excluded because \iiCond is satisfied for these edges.

\subsection*{Fast River}
This thesis is concerned with the computational efficiency of the River voting method for any preference profile, focusing on those resulting in ties in their margin graph. 
To compute River more efficiently in all cases, we also present an algorithm to compute the River diagram for a margin graph given a tiebreak. This has optimizations, which improve the worst-case runtime over that of a naive implementation of River. 
We give the pseudocode for this implementation in \Cref{algo:fastRiver}.

\begin{algorithm}[h]
    \SetAlgoLined
    \DontPrintSemicolon
    \SetKwInOut{Input}{input}\SetKwInOut{Output}{output}
    \SetKwData{hasInEdge}{hasInEdge}
    \SetKwData{riverTrees}{riverTrees}
    \SetKwFunction{add}{add} \SetKwFunction{connected}{sameSet} \SetKwFunction{merge}{union}
    \Input{A margin graph $\mG=(\alts, E)$, a descending linear ordering $\lino \in \allDescOrders$ }
    \Output{The River diagram $\mRV(\lino)$}
    \Begin{
        $\mRV \leftarrow (\alts, \emptyset)$\;
        \hasInEdge $\leftarrow \emptyset$\;
        \riverTrees $\leftarrow \texttt{DisjointSetForest(\alts)}$\;
        \For{$i \leftarrow 1$ \KwTo $\abs{E}$}{
            $(x,y) \leftarrow \lino[i]$\;
            \uIf(\tcp*[f]{Check \iCond}){$y \not \in \hasInEdge$}{
                \uIf(\tcp*[f]{Check \iiCond}){\texttt{not} \riverTrees.\connected($x$,$y$)}{
                    \mRV.\add($(x,y)$)\;

                    \hasInEdge.\add($(x,y)$)\;
                    \riverTrees.\merge($x$,$y$)\;
                }
            }
        }
        \Return{ \mRV}\; 
    }
    \caption{Pseudocode for Fast River \label{algo:fastRiver}}
\end{algorithm}

\begin{lemma}[Fast River correctness]
    \Cref{algo:fastRiver} computes the River diagram from \Cref{def:river} correctly.
\end{lemma}
\begin{proof}
    Let $\mG=(\alts,E)$ be a margin graph of $\prefs \in \allPrefs$ and $\lino \in \allDescOrders$ an ordering of $E$. 
    Note that the pseudocode in \Cref{algo:fastRiver} uses a disjoint-set forest as described in \citeauthor[Ch. 19]{cormenIntroductionAlgorithms2022}. The initialization in line 4 is equivalent to using the $\texttt{MAKE-SET}(a)$ operation for every alternative $a \in \alts$. The other operations are only renamed from \texttt{SAME-COMPONENT} and \texttt{UNION}. 
    To show correctness of \Cref{algo:fastRiver}, we need to show that, for every edge $e \in E$ the River conditions are computed correctly.
    We show \iCond and \iiCond separately.

    For \iCond we show the loop invariant that the set \texttt{hasInEdge} holds the set of alternatives with incoming edges. We initialize the set in line 3 correctly with the empty set, as the initial River diagram has no edges. We assume \texttt{hasInEdge} is correct at the start of an iteration of the loop in line 5. In the loop, once we add an edge to an alternative in the River diagram in line 9, we also add this alternative to the set in line 10. So \texttt{hasInEdge} is also correct after the loop iteration.

    To check whether \iCond is satisfied for the edge $e=(x,y)$ we check membership of $y$ in \texttt{hasInEdge} in line 7 which is equivalent to \iCond from \Cref{def:river} by the invariant we just showed.
    If this is not the case, we proceed to check \iiCond. 
    
    To show that this check is correct, we prove a second loop invariant: 
    The \texttt{riverTrees} disjoint-set forest stores in a set the alternatives that are in the same subtree of the current River diagram.
    We initialize \texttt{riverTrees} with a single set for every alternative in line 4. Since the initial River diagram has no edges, this is correct.
    We assume \texttt{riverTrees} is correct at the start of an iteration of the loop in line 5. 
    If we add an edge $e=(x,y)$ to the River diagram in line 9, we merge the sets containing $x$ and $y$ in line 11. These alternatives are now connected via an edge in the River diagram and thus are part of the same subtree. Thus, \texttt{riverTrees} is correct after the loop iteration.

    To check if \iiCond is satisfied for $e=(x,y)$, we only need to check whether $x$ and $y$ are in the same River subtree in line 8. This is because we already know $y$ has no incoming edge from line 7. Thus, $y$ must be the root of any subtree it is in. 
    So if $x$ and this $y$ are in the same subtree, there is a path from $y$ to $x$ in the River diagram. Thus, we check \iiCond from \Cref{def:river} correctly in line 8.

    Since we iterate over the edges in $E$ by the linear descending order of \lino in the loop in line 5, check \iCond and \iiCond correctly in line 7 and line 8 and only add an edge in line 9 if neither is satisfied, we correctly compute the River diagram $\mRV(\lino)$ and return it in line 13.

\end{proof}
We now prove the runtime of this algorithm.
\begin{lemma}[Fast River runtime]
    \label{lma:rvInP}
    Given a margin graph $\mG =(\alts, E)$ for a preference profile $\prefs \in \allPrefs$ and an ordering $\lino \in \allDescOrders$,
    the runtime of computing the River diagram is in $\bigO{n^2 \cdot \alpha(n)}$ using \Cref{algo:fastRiver}, where $\alpha$ is the inverse ackermann function and $n$ is the number of alternatives.
\end{lemma}
\begin{proof}
    Let $\mG=(\alts,E)$ be a margin graph of $\prefs \in \allPrefs$ and $\lino \in \allDescOrders$ an ordering of $E$. Let $n=\abs{\alts}$.
    The initialization of the empty River diagram in adjacency list representation takes $\bigO{n}$, the initialization of the \texttt{hasInEdge} set is in $\bigO{m}$ and the initialization of the disjoint-set forest is done with $n$ $\texttt{MAKE-SET}$ operations.
    The loop from line 5 to line 12 is iterated $\abs{E}$ times.
    In the loop, checking set membership of an edge in line 7 and adding an edge to a set in lines 9 and 10, are in $\bigO{1}$.
    We also perform at most one \texttt{SAME-COMPONENT} and one \texttt{UNION} operation. The \texttt{SAME-COMPONENT} operation consists of two \texttt{FIND-SET} operations.

    The runtime of these $\abs{E}$ \texttt{UNION}, $2 \cdot \abs{E}$ \texttt{FIND-SET} and $n$ $\texttt{MAKE-SET}$ operations is in $\bigO{(3 \cdot \abs{E} + n) \cdot \alpha(n)}$ using a disjoint-set with union by rank and path compression \cite[Theorem 19.14]{cormenIntroductionAlgorithms2022}. 
    This dominates the total runtime.
    Observe that \mG is a margin graph with at most $2 \cdot n \cdot (n-1)$ edges. Thus, the total runtime is in $\bigO{n^2 \cdot \alpha(n)}$. 
\end{proof}
If no ordering is given, the edges need to be sorted by margin to create an ordering. In that case, River has a worst-case runtime in $\bigO{n^2 \cdot \log(n)}$.

    \chapter{Results}

    \section*{Overview}
We show that deciding whether an alternative \wcand is an $\RVPUT$ winner is computable in polynomial time. From this follows that computing the set of winning alternatives is also computable in polynomial time.
We can first check for Condorcet winners and terminate if there is one. 
Otherwise, we give a process that can be used to find a tiebreak for an alternative \wcand such that they win River with that tiebreak, if and only if they are an $\RVPUT$ winner:
\begin{enumerate}
    \item Create a semi-River diagram for the margin graph (\Cref{sct:1}).
    \item Find the \rsptf for an alternative on this semi-River diagram (\Cref{sct:2}).
    \item Generate a tiebreak ordering based on this \rsptf (\Cref{sct:3}).
    \item Run River with this tiebreak (\Cref{sct:mainProof}). 
\end{enumerate}
If $w$ is an $\RVPUT$ winner, the process also results in a River diagram where $w$ wins.
We define this process, namely the \rvPutCheck, in \Cref{sct:mainProof} and prove its correctness.
By running the \textsc{Constructive $\RVPUT$ Check} for every alternative, the set of winning alternatives under $\RVPUT$ can be computed.

To give an intuition, this process works because the edges remaining in the semi-River diagram have specific properties. The tree we compute for a selected alternative on this semi-River diagram finds the strongest paths. Through this combination the tree is a possible River diagram, if the alternative is an $\RVPUT$ winner. To prove this, we construct a tiebreak with which River generates this tree and execute River with this tiebreak to check whether it results in this tree. For an $\RVPUT$ looser River cannot generate a tree where the looser wins. Thus, we can check if our selected alternative wins River with the constructed tiebreak to decide whether it is a winner of River with PUT.

    \section{The Semi-River Diagram} \label{sct:1}
We create the semi-River diagram, to only keep edges with guaranteed properties. For a preference profile \prefs, this diagram is a subgraph of $\mG(\prefs)$.
Essentially, the semi-River diagram contains all edges that are in a River diagram under some tiebreak. 

\subsection*{The Semi-River Process}
    We create this diagram with the semi-River process.
    It iterates over the edges of the margin graph $\mG(\prefs)$ in the order of any $\lino \in \allDescOrders$. For an edge $e=(x,y)$ the semi-River process checks if it should include $e$, by having two conditions similar to the River branching condition and the River cycle condition. These are the semi-River branching condition \sRVbranchingCond and the semi-River cycle condition \sRVcycleCond. 
    The semi-River conditions are more general than the River conditions, in the sense that \sRVbranchingCond and \sRVcycleCond respectively test whether \iCond or \iiCond are satisfied for edge $e$ in every universe.
    The \sRVbranchingCond does this, by looking if there is a strictly higher incoming edge of $y$ in any River diagram before $e$ is processed.
    The \sRVcycleCond does this, by looking if there is a path from $y$ to $x$ in any River diagram before $e$ is processed.

    Like with the River diagram, we call a semi-River diagram after processing the $i$-th edge $\msRVi$. We call a semi-River diagram after processing all edges with margin $k$ or higher $\msRVm[k]$. We denote the semi-River diagram after processing only edges with strictly higher margin by $\msRVgm[k]$.
    For the semi-River diagram after processing all edges $\msRVi[\abs{E(\mG)}]$ or $\msRVm[\min{\msset{\m{e} \mid e \in E(\mG)}}]$, we write \msRV.
\begin{definition}[The semi-River process]\label{def:strongSemiRiverProcess}
    \begin{enumerate}
    \item Create an arbitrary descending ordering $\lino \in \allDescOrders$ of $E(\mG)$. \label{step:srv:chooseOrdering}
    \item Initialize the semi-River diagram $\msRVi[0]$ as a directed graph $(\alts, \emptyset)$.
    \item For $i$ from $1$ to $\abs{E(\mG)}$, process the edge $e =(x,y) := \lino[i]$ and add $e$ to \msRVi[i], unless the \textbf{semi-River branching condition \sRVbranchingCond} or the \textbf{semi-River cycle condition \sRVcycleCond} is satisfied:
        \item[]
            \begin{tabular}{rp{0.75\linewidth}}
            \textbf{\sRVbranchingCond}  & There is an other incoming edge to $y$ namely $e'=(z,y)\in E(\msRVgm[\m{e}])$, 
                                         with no path from $y$ to $z$ in $\msRVm[\m{e'}]$.\\
            \textbf{\sRVcycleCond}      & The ancestors of $x$ up to $y$ in $\msRVgm[\m{e}]$ form a DAG in $\msRVgm[\m{e}]$ where $y$ is the only source.
            \end{tabular}
    \end{enumerate}
\end{definition}

To better understand how the semi-River process behaves, we now give some examples.
For uniquely weighted margin graphs, the semi-River diagram is equal to the River diagram on that margin graph. This is because with no tied edges, the \sRVbranchingCond does not allow any alternative to have more than one incoming edge and so the \sRVcycleCond just checks for a path from $y$ to $x$.
If in the margin graph all edges are tied edges then all edges of the margin graph are in the semi-River diagram. This is because the conditions can only be satisfied through strictly higher edges.
If there are two tied incoming edges for one alternative in the margin graph, both of which might be in a River diagram for some tiebreak then they are both in the semi-River diagram.

To be able to omit the linear order chosen in Step 1 in future arguments we prove the following lemma:
\begin{lemma}
    For any descending linear order $\lino \in \allDescOrders$ chosen in Step 1 of \Cref{def:strongSemiRiverProcess}, the semi-River diagram resulting from the algorithm is the same.
\end{lemma}
\begin{proof}
    Let $\lino, {\lino}{'} \in \allDescOrders$ be different descending linear orderings.
    For two edges  $e_>,e_< \in E(\mG)$ with different margins $\m{e_>} > \m{e_<}$, the edge $e_>$ is before $e_<$ in both $\lino$ and ${\lino}{'}$.
    Therefore, $\lino$ and $\lino'$ can only differ in the order within groups of edges with the same margin.
    
    When considering $e \in E(\mG)$, the two semi-River conditions only perform checks using the edges with strictly higher margin in $\msRVgm[\m{e}]$.
    So, different orders of edges with the same margin as $e$ in $\lino$ and ${\lino}{'}$ have no impact on whether or not $e$ is added to the semi-River diagram.
    Thus, this decision is the same for both $\lino$ and ${\lino}{'}$, and we conclude that the semi-River diagrams using $\lino$ and ${\lino}{'}$ are identical.
\end{proof}

We prove the runtime of the semi-River process is polynomial.
\begin{theorem}\label{thm:srvInP}
    The semi-River diagram can be computed with a worst-case runtime polynomial in $\abs{\alts}$.
\end{theorem}
\begin{proof} 

    Let $\mG=(\alts, E) $be the margin graph of any preference profile \prefs.
    In this proof we write $n$ for $\abs{\alts}$. 

    First, we state how the conditions may be computed.
    Let $e =(x,y) \in E$ be the edge for which the conditions are checked.
    The \sRVbranchingCond is equivalent to checking for every incoming edge $(z,y)$ to $y$ in \msRV with a strictly higher margin than $\m{e}$ whether $z$ is in $\text{BFS}(y)$ on \msRVgm[\m{z,y}]. If there is such an edge, the \sRVbranchingCond is satisfied.
    To compute the \sRVcycleCond we first compute the transposed graph ${\msRVgm[\m{e}]}^{-1}$ where the transposed outgoing edges of $y$ are removed and call it $G'$. We then do a BFS from $x$ on $G'$ and save the result in the set of alternatives $A'$. We then compute the subgraph induced by $A'$ on $\msRVgm[\m{e}]$ and check if it is a DAG where $y$ is the only source. Since $A'$ matches the definition of the ancestors of $x$ up to $y$, checking whether the  subgraph of $\msRVgm[\m{e}]$ induced by them is a DAG is the same as checking \sRVcycleCond.
    Note that we can maintain the transposed graph of \msRV in the loop by also adding the reversed edge of every edge we add to the initially empty semi-River diagram to the initially empty transposed graph.
    Also, note that by allowing a BFS to only take edges with a margin higher than $\m{e}$ doing a BFS on $\msRVgm[\m{e}]$ given a semi-River diagram where edges with margin $\m{e}$ are also included has the same runtime as doing it on that diagram. 

    Now we argue that the runtime of the entire semi-River process is polynomial.
    Step 1 is in $\bigO{\abs{E} \cdot \log \abs{E}}$, by sorting the edges in $E$.
    In Step 2 the initialization of the graph and the transposed graph is in $\bigO{n}$ if use an adjacency list representation.
    The loop in Step 3 is where the process iterates over the $\abs{E}$ edges. We write \msRV and ${\msRV}^{-1}$ for the intermediate diagrams in that loop.

    For an edge $e=(x,y) \in E$, we check both conditions as described above. 
    To check \sRVbranchingCond, we iterate all at most $n-1$ edges $(y,z)$ outgoing from $y$ in ${\msRV}^{-1}$. For every such edge $(y,z)$ we compute a BFS from $z$ on ${\msRVgm[\m{y,z}]}^{-1}$ in $\bigO{n + \abs{E}}$ once and check if $y$ is in the resulting set in $\bigO{1}$. Thereby checking if \sRVbranchingCond is satisfied for $e$ takes $\bigO{n \cdot (n + \abs{E})}$ time.

    To check \sRVcycleCond we compute $A'$ by also ignoring the outgoing edges of $y$ in the BFS from $x$ on $\msRVgm[\m{e}]^{-1}$. Creating a set of the outgoing edges of $y$ is in $\bigO{n}$. Then doing the BFS that does not use these edges is in $\bigO{n + \abs{E}}$. Now we compute the subgraph induced by $A'$ on $\msRVgm[\m{e}]$ in $\bigO{n + \abs{E}}$. Finally, we check whether this subgraph is an acyclic graph in $\bigO{n + \abs{E}}$ \cite[Ch. 22.4]{cormenIntroductionAlgorithms2022}. 
    To check if $y$ is the only alternative with no incoming edges in this subgraph takes $\bigO{\abs{E}}$, by checking for all edges if they are incoming to $y$.
    Thereby, checking if \sRVcycleCond is satisfied for $e$ takes $\bigO{n + \abs{E}}$ time. 

    Inserting an edge into ${\msRV}$ and ${\msRV}^{-1}$ takes $\bigO{1}$.

    Observe that \mG is a margin graph with at most $2 \cdot n \cdot (n-1)$ edges.
    We execute the loop $\abs{E}$ times, which is in  $\bigO{n^2}$.
    Combining the operations in step 3 yields a runtime in $\bigO{n^2 \cdot (n^3 + n^2)}$, which is in $\bigO{n^5}$.
    So runtime of the semi-River process is in $\bigO{n^5}$, and thus polynomial in the input size.

\end{proof}

\subsection*{The Semi-River Diagram Includes All Edges of Any River Diagram}
For a preference profile $\prefs \in \allPrefs$, we use the semi-River diagram of the margin graph $\mG(\prefs)$ in further algorithms to decide the $\RVPUT$ winners of \prefs. We now analyze how the semi-River diagram relates to any River diagram of \prefs.

The \Cref{lma:semiRiverPathExistance} shows an essential property of the semi-River diagram, which we use in later sections. It is illustrated in \Cref{fig:semiRiverPathLemma}.
\begin{restatable}[]{lemma}{semiRiverPathExistance}
    \label{lma:semiRiverPathExistance}
    Let \msRV be the semi-River diagram on $\mG(\prefs)$ for $\prefs \in \allPrefs$. For an alternative $y \in \alts$ and any two incoming edges $e, e' \in E(\msRV)$ with $e=(x,y)$ and $e'=(z,y)$, if $\m{e}<\m{e'}$, then there is a path \Path{z}{y} in $\msRVgm[\m{e'}]$. 
\end{restatable}
\begin{proof}
    Let $y\in \alts$ have incoming edges $e, e' \in E(\msRV)$ with $e=(x,y)$ and $e'=(z,y)$ and $\m{e}<\m{e'}$.
    Assume towards contradiction that there is no path \Path{z}{y} in \msRV. Then the \sRVbranchingCond would be true for $e$ through the edge $e'$, which contradicts our assumption that $e \in E(\msRV)$.
\end{proof}
\begin{figure}[t] \centering
    \includegraphics[width = 0.5\linewidth]{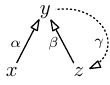}
    \caption{The edges $(x,y) \in \msRV$ with margin $\alpha$ and $(z,y) \in \msRV$ with margin $\beta$. Note that $\alpha < \beta$. There is the path from $y$ to $z$ in \msRV with strength $\gamma$ and $\beta \geq \gamma$. }
    \label{fig:semiRiverPathLemma}
\end{figure}

In order to relate the semi-River diagram back to River, we prove it is a superset of all River diagrams.
\begin{theorem} \label{pps:msRVsupersetMRV}
For any preference profile \prefs, an edge that is in a River diagram with some tiebreak is also in the semi-River diagram. Formally:
   $$E(\msRV(\prefs)) \supseteq \bigcup_{\lino \in \allDescOrders } E(\mRV(\lino, \prefs))$$
\end{theorem}
\setcounter{claim}{0}
\begin{proof}
    Let $\prefs$ be a preference profile with the margin graph \mG and the semi-River diagram \msRV. We denote the margin graph edges $E(\mG)$ as $E$, and the semi-River edges $E(\msRVm)$ as $\EsRVm$.
    We prove the contraposition of \Cref{pps:msRVsupersetMRV}: For an edge $e\in E$ and  
    for all $\lino \in \allDescOrders$, if $e\notin \EsRV$, then $e\notin E(\mRV(\lino))$.
    We show this via induction over the edges in $E$ in order of decreasing margin.
    Let $\lino \in \allDescOrders$. Denote the River edges $E(\mRV(\lino))$ as $\ERV$.

    For the induction base, let $k_{\max}$ be the maximum margin of all edges in $E$. For an edge with this margin in the margin graph \sRVbranchingCond and \sRVcycleCond cannot be satisfied, because there exists no edge with higher margin in the semi-River diagram. So all edges with maximum margin are included in $\EsRVm[k_{\max}]$.
    Therefore, the implication is true that all such edges that are not in $\EsRVm[k_{\max}]$ are also not in $\ERV$.
 
    As our induction hypothesis for all edges $e_>\in E$ with $\m{e_>}> k$, we assume that $e_> \notin \EsRVgm[k]$ implies $e_> \notin \ERV$.
    From this induction hypothesis follows \Cref{315:claim:pathExcludedByIH}, which will be useful in the induction step.
    \begin{claim}\label{315:claim:pathExcludedByIH}
        For any path \PathP in the margin graph \mG, holds that if \PathP is not in $\msRVgm[k]$, then it is also not in $\mRV(\lino)$. 
    \end{claim} 
    \begin{proof}
        Since \PathP is not in $\msRVgm[k]$, there is at least one edge on \PathP which is not in $\EsRVgm[k]$. By the induction hypothesis, this edge is also not in $\ERV$. Thus \PathP is not in $\mRV(\lino)$.
    \end{proof}

    For the induction step, let $e=(x,y) \in E$ with $\m{e}=k$ and $e \notin \EsRVm[k]$. This means either \sRVbranchingCond or \sRVcycleCond are satisfied for $e$ on $\EsRVgm[k]$, by \Cref{def:strongSemiRiverProcess}.
    We distinguish both cases and show $e \notin \ERV$.
    \begin{casesproof}
        \item
    \textbf{Let \sRVbranchingCond be satisfied for $e=(x,y)$.} \quad 
    In this case, we show that the River branching condition \iCond is satisfied for $e$ via another incoming edge which is added to the River diagram before $e$.
    From our assumption that \sRVbranchingCond is satisfied for $e$ follows that there exists an incoming edge $e' = (z,y)$ of $y$ in $\msRVgm[\m{e}]$ with no path \Path{y}{z} in $\msRVm[\m{e'}]$.
    We show that such an $e'$ also exists in \ERV.
    We distinguish the cases whether or not \iCond is satisfied for $e'$,
    as in either case there is an incoming edge to $y$ when $e$ is processed by River.
    \begin{casesproof}
    \item
        If the \iCond is satisfied for $e'$, then there is some edge $e''$ incoming to $y$ already in $\ERV$.
        Because it has higher margin than $e$, the edge $e'$ is before $e$ in $\lino$ and thus it is processed by River first.
        So in this case, $e''$ also satisfies \iCond for $e$.
    \item
        Otherwise, \iCond is false for $e'$.
        Since the path \Path{z}{y} is not in $\msRVgm[\m{e}]$, it follows from \Cref{315:claim:pathExcludedByIH}, that \Path{z}{y} is not in $\mRV(\lino)$, and so \iiCond is also false for $e'$.
        Thus $e'$ is added to $\ERV$ and satisfies \iCond for $e$.
    \end{casesproof}

    So $e \notin \ERV$, because \iCond is satisfied for $e$ when it is processed by River.
        \item
    \textbf{Let \sRVcycleCond be satisfied for $e=(x,y)$.} \quad
    In this case, we show that the River cycle condition \iiCond is satisfied for $e$. 
    Using our assumption that \sRVcycleCond is satisfied for $e$ we find a path \Path{y}{x} in \msRVgm[\m{e}] that is in $\mRV(\lino)$. 
    Recall by the \sRVcycleCond we know that all ancestors of $x$ up to $y$ from a DAG in \msRV where $y$ is the only source. 

    To find \Path{y}{x} we start in $x$ and select the first incoming edge $(z,x)$ in $\lino$ in $\msRVgm[\m{e}]$ to the predecessor $z$. We recursively repeat this from $z$ to its predecessor $z'$, until we reach $y$.
    For every edge $(u,v)$ on this \Path{y}{x}, 
    $(u,v)$ is the first incoming edge to $v$ in \lino, the \iCond is false for $(u,v)$. 
    Also, for every path \Path{v}{u} from $v$ to $u$ in the margin graph, \Path{v}{u} is not in $\msRVm[\m{u,v}]$, as that would create a cycle and contradict that all ancestors of $x$ up to $y$ form a DAG. This \Path{v}{u} is not in $\ERV$ by \Cref{315:claim:pathExcludedByIH} and thus \iiCond is false for $(u,v)$.
    So \iCond and \iiCond are false for $(u,v)$, therefore we conclude $(u,v) \in \ERV$. 
    Since all edges on \Path{y}{x} are in $\ERV$, the path \Path{y}{x} is in $\mRV(\lino)$. Thus \Path{y}{x} satisfies \iiCond for $e$.

    So $e \notin \ERV$, because \iiCond is satisfied for $e$ when it is processed.
\end{casesproof}

    With this case distinction we show that if \sRVbranchingCond is satisfied for $e$, this implies that \iCond is satisfied for $e$, and if \sRVcycleCond is satisfied for $e$, this implies that \iiCond is satisfied for $e$. So if $e$ is excluded from the semi-River diagram, it is also excluded from the River diagram, for the arbitrary ordering $\lino$.

    This shows our induction step and we conclude that for any preference profile, for any edge $e\in E$, and any $\lino \in \allDescOrders$, if $e\notin \EsRV$, then $e\notin E(\mRV(\lino))$.
    This is the contraposition of \Cref{pps:msRVsupersetMRV} and thus concludes the proof.
\end{proof}

Note that this proof shows the intuition we gave at the beginning of this section, namely, \sRVbranchingCond implies \iCond and \sRVcycleCond implies \iiCond for any edge in the margin graph.

    \section{The Recursive Strongest Path Tree} \label{sct:2}
In this section, we define the \rsptf and give the \ourPrim algorithm to compute it. Its essential property is, that if a vertex precedes another vertex via a path in the tree, then this path is a \textit{recursive strongest path} between them in the graph the tree was computed on. A \textit{recursive strongest path} is a path, where all subpaths are also strongest paths.
We use this property to show the correctness of our process.

This section is independent of social choice theory, so we use definitions in terms of graph theory. These definitions are analogous if graph $G$ is a margin graph.
\subsection*{Defining the Recursive Strongest Path Tree}

In a weighted directed graph $G=(V,E)$ with $w \colon E \to \mathbb{N}$, the \textit{strength} of a path $\Path{v_1}{v_l}=\langle v_1, \dots, v_l \rangle$ is the weight of its smallest edge. 
We define the \textit{strongest path weight} $\minweight(v_1,v_l)$ from $v_1$ to $v_l$ in $G$ as follows: $$ \minweight(v_1,v_l) = \begin{cases*}
    \max\sset{\str(\PathP) \mid \PathP \text{ any path from $v_1$ to $v_l$ in $G$ }} & \text{if there is a path from $v_1$ to $v_l$,}\\
    -\infty & \text{otherwise.}
\end{cases*}$$
A \textit{strongest path} from vertex $u$ to vertex $v$ is any path $\Path{u}{v}$ with $\str(\Path{u}{v}) = \minweight(u,v)$.

The property that all subpaths of a strongest path are also strongest paths is essential for the concept of the \rsptf. We formally define this for paths first.
\begin{definition}[recursive strongest path]\label{def:recursiveStrongestPath}
A \textit{recursive strongest path} from vertex $v_1$ to vertex $v_l$ is any path $\Path{v_1}{v_l}$ where all subpaths are also strongest paths. That is, for all $1\leq i < j \leq \abs{\Path{v_1}{v_l}}$ the subpath $\Path{v_1}{v}[i:j] = \langle v_i, v_{i+1}, \dots, v_j\rangle$ is a strongest path from $v_i$ to $v_j$ in $G$. Formally: $\str(\Path{v_1}{v_l}[i:j]) = \minweight(v_i,v_j)$.
\end{definition}
It is easy to see that all subpaths of a recursive strongest path are also recursive strongest paths.
The intuition behind this definition is, that we want optimal substructure for strongest paths, which they usually lack.
From this follows the definition of the \rsptf.
\begin{definition}[\rsptf] \label{def:rspt}
    A \textit{\rsptf} $T$ of $G=(V,E)$ is a tree rooted in $s \in V$.
    For vertices $u,v\in V$ with $u \prec v$ in $T$ the unique simple path from $u$ to $v$ in $T$ is a recurse strongest path.
    Note that since $s$ is the root it has recursive strongest paths to all vertices in $V$.
\end{definition}
Now we give an algorithm to compute a \rsptf with a designated root for a given graph.

\subsection*{Applying Prim's Algorithm to a Directed Graph}
Prim's algorithm is usually used to find a minimum spanning tree (MST) in an undirected weighted graph.
It operates by selecting a start vertex arbitrarily and then growing the set of vertices in the MST from there, by always selecting the minimum weight outgoing edge, until the MST spans the entire graph
\cite[Ch. 21]{cormenIntroductionAlgorithms2022}, \cite{primShortestConnectionNetworks1957}.
We apply Prim's algorithm on a directed graph. Note that this does not result in an MST. Here the selected start vertex becomes the root of the tree. 
Instead of the minimum outgoing edge, we use the maximum outgoing edge. 
So, we run Prim's Algorithm, rooted in vertex $s$ in a directed graph.
We call this the \ourPrim algorithm and it computes a \rsptf in $G$ rooted in $s$.
To give an intuition of why this computes such a tree:
The \ourPrim algorithm always takes the maximum edge, which connects the current tree to a vertex that is not already in the tree. With this greedy approach, it always picks edges that are part of some strongest path. Regard the vertices $u,v\in V$, with $u$ already in the tree and $v$ not. If $(u,v)$ is the strongest edge, which connects the current tree to any vertex not yet in the tree, then there can be no stronger path from $u$ to $v$ than the edge $(u,v)$.

We now give a more formal description. Note that this is a more abstract description of Prim's algorithm than that given by \textcite{cormenIntroductionAlgorithms2022} which is more suited to our purpose.
\begin{definition}[The \ourPrim algorithm]
    \label{rspt:algorithm}
Given a weighted directed graph $G=(V,E)$ with weight function $w \colon E \to \mathbb{N}$ and a start vertex $s\in V$, which has a path in $G$ to every other vertex in $V$.
\begin{enumerate}
    \item Initialize $S$ with $\sset{s}$ and $T$ with $(V, \emptyset)$.
    \item while the set of crossing edges $\msset{(u,v) \in E \mid u\in S \land v \not \in S}$ is not empty:
    \begin{enumerate}
        \item[2.1.]  extract any maximum weighted edge $e_{\text{max}}=(u,v)$ from the set of crossing edges.
        \item[2.2.]  add $v$ to $S$ and add $e_{\text{max}}$ to $E(T)$.
    \end{enumerate}
    \item return $T$.
\end{enumerate}
In the remainder of the thesis, we write $\rspt(s, G)$ for the result of the \ourPrim algorithm with input $G=(V,E)$ and $s\in V$.

In reference to this algorithm for the remainder of the thesis, we call the vertices in $S$ \textit{explored} and the remaining vertices in \Vrem the set of \textit{unexplored} vertices. We call an edge $(u,v) \in E$ with $u \in S$ and $v \in \Vrem$ a \textit{crossing} edge.
The \ourPrim algorithm consists of repeatedly adding a \textit{maximum crossing edge}. This is any crossing edge with maximum weight among all crossing edges. 
\end{definition}
Note that we restrict the starting vertex $s$ to vertices with a path to every other vertex in $V$, as in the remainder of the thesis, the resulting \rspt is only relevant for our decision process if $s$ has paths to every other vertex $v\in V$ in $G$. If we execute \ourPrim with other starting vertices the only properties we need are that it has polynomial runtime and results in any subgraph of $G$. The latter property is evident as it only adds vertices and edges from $E(G)$.

We prove correctness of the \ourPrim algorithm next.
\begin{theorem} \label{pps:AlgoRSPT}
   Let $G=(V,E)$ be a directed graph with weight function $w \colon E \to \mathbb{N}$ and a start vertex $s\in V$ with a path to every other vertex in $V$. \ourPrim started on $s$ results in a \rsptf of $G$ rooted in $s$.
\end{theorem}
\setcounter{claim}{0}
\begin{proof}
    Let $G=(V,E)$ be the input graph with start vertex $s\in V$ which has a path to every other vertex in $V$. In this proof, we write \rspt for $\rspt(s, G)$.

    We first prove the loop invariants that (\RN{1}) $T$ is a forest and (\RN{2}) $S$ is the set vertices that have a path from the root $s$ and (\RN{3}) the vertices in $V \setminus S$ are isolated in $T$.
    At the initialization in step 1 the invariants hold as all vertices of $V$ are isolated in $T$ and $S=\sset{s}$.
    We assume the invariants hold for $T$ and $S$ at the start of the loop in step 2.
    The edge $(u,v)$ extracted in step 2.1 is a crossing edge. This means $u \in S$ and $v \notin S$.
    In step 2.2 the edge $(u,v)$ is added to $E(T)$.
    After this, $T$ is still a forest, because the previously isolated vertex $v$ now has exactly one incoming edge in $E(T)$ and cannot be part of a cycle in $T$. Thus, (\RN{1}) holds after the iteration of the loop.
    In step 2.2. the vertex $v$ is also added to $S$. The path from $s$ to $u$ in $T$ combined with $(u,v)\in E(T)$ is a path from $s$ vertex $v$ in $T$. Thus, (\RN{2}) holds after the iteration of the loop. The vertex $v$ is no longer in $V \setminus S$ and thus (\RN{3}) also still holds.
    The loop terminates because the loop condition in step 2 can only be true if there is a vertex in $V \setminus S$, but at every iteration, a vertex from $\Vrem$ is added to $S$ in step 2.2 and $V$ is finite.
    Since the algorithm only terminates once there is no more crossing edge and by our restriction there is a path from $s$ to every vertex in $G$ at the end $S=V$.

    From these invariants follows that \rspt is a tree.

    For any two vertices with a path from one to the other in \rspt, we show that there exists no stronger path in $G$ by contradiction.
    Let $u,v \in V$ with $u \prec v$ via the path \Path{u}{v} in \rspt.
    From  $u \prec v$ follows that \ourPrim adds $u$ to $S$ before $v$, because all edges on \Path{u}{v} must be crossing edges at some point in order to be selected.
    We assume towards contradiction that there is a stronger path $\Path{u}{v}'$ with $\str(\Path{u}{v}') > \str(\Path{u}{v})$.
    On the path $\Path{u}{v}$ in $\rspt$ there is a first minimum edge $e_{\min}$ with $w(e_{\min})=\str(\Path{u}{v})$ which is first selected by \ourPrim, by the definition of the path $\str$. 
    From our assumption that $\str(\Path{u}{v}') > \str(\Path{u}{v})$, follows for every edge $e'$ on $\Path{u}{v}'$ that $w(e') > w(e_{\min})$.
    We now regard the point during the algorithm's execution where $e_{\min}$ is a crossing edge that is selected by \ourPrim.
    For an illustration, regard \Cref{fig:rsptProofIntuition}.
    \begin{figure}[t]
        \centering
        \includegraphics[width = 0.5 \linewidth]{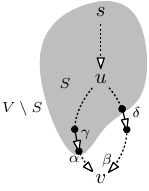}
        \caption{An illustration of \msRV for the proof of \Cref{pps:AlgoRSPT}. We see that, since $u\in S$ and $v \in \Vrem$, there are crossing edges on the path \Path{u}{v} from $u$ to $v$ with strength $\alpha$ which is \rspt and the supposedly stronger path $\Path{u}{v}'$ in $G$ with strength $b$. The edge $e_{\min}$ on \Path{u}{v} has weight  $\gamma$ with $\gamma = \alpha$ and the crossing edge on $\Path{u}{v}'$ has weight $\delta$ with $\delta \geq \beta$. The assumption in the proof is that $\alpha < \beta$ and thus $\gamma<\delta$.}
        \label{fig:rsptProofIntuition}
    \end{figure}
    Recall that the \ourPrim algorithm always selects the maximum crossing edge. 
    Since $\Path{u}{v}'$ starts in $u \in S$, there is an edge $e'$ on $\Path{u}{v}'$ which is a crossing edge.
    That $e_{\min}$ on \Path{u}{v} in \rspt is selected instead of $e'$ contradicts that $w(e')> w(e_{\min})$.
    Thus the assumption that a stronger path $\Path{u}{v}'$ exists is false.

    So for all $u,v \in V$ with $u \prec v$, the path \Path{u}{v} in \rspt is a strongest path.
    As we picked $u$ and $v$ arbitrarily, the path \Path{u}{v} is also a recursive strongest path.
\end{proof}

Since \ourPrim is simply the application of the Prim algorithm to a directed graph, 
we can use Prim as described in \textcite[Ch 21.2]{cormenIntroductionAlgorithms2022} only with directed edges, a fixed root, and a maximum priority queue and it outputs a tree with the same properties. That algorithm can be implemented to have a runtime in $\bigO{\abs{E} + \abs{V} \log \abs{V}}$ \cite{cormenIntroductionAlgorithms2022}. 
It follows that:
\begin{corollary}
    \label{obs:rsptInP}
    For input $G=(V,E)$ and $s \in V$ computing \ourPrim is polynomial in the input size.
\end{corollary}

\section{Introducing A Tiebreak Based On Selected Edges} \label{sct:3} 

To propagate the information we gained by creating the \rsptf for an alternative on the semi-River diagram to the River process, we devise a linear descending ordering over the edges of the margin graph i.e. a tiebreak. We execute River with this ordering to check if the alternative is a winner.
To define this special ordering in the next section, we now give a way to compute an ordering intuitively favoring a set of selected edges.
This is done by ordering the selected edges before the edges that are not selected while maintaining a descending order by margin.
Note that for an arbitrary edge in the margin graph, being in this selected set does not mean that it is the River diagram with the resulting ordering.

\begin{definition} \label{def:ordering}
For a preference profile \prefs with margin graph \mG, with edges $E(\mG)$ and a subset $E' \subseteq E(\mG)$, we define an ordering $o^{E'} \in \allDescOrders$, such that the edges are ordered by decreasing margin and sets of edges with the same margin are ordered based on their inclusion in $E'$: 
\begin{enumerate}
    \item[(1)] $\m{e} > \m{e'} \rightarrow \text{inv}_{o^{E'}}(e) <\text{inv}_{o^{E'}}(e')$ and 
    \item[(2)] $ \m{e} = \m{e'} \land 
    e\in E' \land e' \not \in E' \rightarrow\text{inv}_{o^{E'}}(e) <\text{inv}_{o^{E'}}(e')$.
\end{enumerate}
We define the function $\setOrdering$ that maps from $E(\mG)$ and $E'$ to this $o^{E'}$, with $\setOrdering(E(\mG),E') = o^{E'}$.
\end{definition}
To give some examples of such an ordering: If $E'=E(\mG)$ or $E'=\emptyset$, then the resulting ordering can be any ordering from $\allDescOrders$, because the property (2) would trivially be true. If $E'$ contains only a single edge, this edge is first among all edges with the same margin. If $E'$ contains multiple edges of the same margin, then these are before the other edges with the same margin in the resulting descending linear order but have arbitrary order among themselves.

To compute $\setOrdering$ this, we first sort the edges in $E(\mG)$ by their inclusion in $E'$ to satisfy (2) and then sort them stably by their margin in descending order to satisfy (1). 
As we use the $\setOrdering$ in our process, we need to show that it has polynomial runtime. 
\begin{lemma}
    \label{pps:oInP}
    Given the margin graph \mG of preference profile $\prefs \in \allPrefs$ for $E'\subseteq E(\mG)$, computing the $\setOrdering$ has polynomial runtime.
\end{lemma}
\begin{proof}
    The $\setOrdering$ can be computed by sorting the edges in $E(\mG)$ twice. The comparison operation for either sorting can be done in $\bigO{1}$.
    Since sorting is in $\bigO{\abs{E(\mG)} \cdot \log (\abs{E(\mG)})}$ this is also the runtime of computing $\setOrdering$. 

\end{proof}

\section{The Algorithm to Compute River PUT Winners} \label{sct:mainProof}

We combine the algorithms from \Cref{sct:1}, \Cref{sct:2}, and \Cref{sct:3} with River into the \rvPutCheck Algorithm:
For a margin graph of a preference profile, we first compute the semi-River diagram. On this diagram, we compute the \rsptf for the alternative we want to check. Based on the edges of this tree we construct an ordering with which River outputs a diagram, where our alternative wins, if and only if it is a $\RVPUT$ winner. 
We recall the notation from the previous sections and then combine them to formally define this \rvPutCheck.

\subsection*{Building the \rvPutCheck}
\subsubsection*{The semi-River Diagram from \Cref{sct:1}}
The semi-River diagram $\msRV(\prefs)$ is a subgraph of the margin graph \mG of the preference profile \prefs. We define it in \Cref{def:strongSemiRiverProcess}. It contains all edges that are part of a River diagram under some tiebreak as shown in \Cref{pps:msRVsupersetMRV}.
We restate \Cref{lma:semiRiverPathExistance} here as we use it extensively in the later proof:
\semiRiverPathExistance*


\subsubsection*{The \rsptf from \Cref{sct:2}}
Taking this semi-River diagram \msRV, we construct a \rsptf for the alternative $\wcand$ in \msRV.
The $\rspt(\wcand, \msRV)$ is such a tree. It is the result of the \ourPrim algorithm. This algorithm always picks a maximum edge crossing from the set of explored nodes to the set of unexplored nodes. 
Note that alternatives that are $\RVPUT$ winners have immunity by \Cref{lma:PUTwinnersBreakAllIncomingEdges}. So, these alternatives meet our precondition from \Cref{rspt:algorithm}. For other alternatives, the output graph is not relevant. 
We recall the following properties of $\rspt(\wcand, \msRV)$ from \Cref{def:rspt} on the semi-River diagram:
\begin{itemize}
    \item \rspt is a tree. This implies that every alternative has at most one incoming edge and there are no cycles. 
    \item The alternative $\wcand$ is the root of \rspt.
    \item For all alternatives $u,v \in \alts$, if $u \prec v$ this implies that the unique path from $u$ to $v$ in \rspt is a recursive strongest path in the semi-River diagram \msRV. 
\end{itemize}
In this section, we show $\rspt(\wcand, \msRV)$ is a River diagram if and only if \wcand is a $\RVPUT$ winner.

\subsubsection*{The $\setOrdering$ from \Cref{sct:3}}
We use $\setOrdering$ to create a tiebreak where \wcand may be a winner of $\RVPUT$ for a preference profile \prefs.
We name this ordering $\ospecial$ and define it as 
\begin{definition}
\[\ospecial = \setOrdering(E(\mG),E(\rspt(\wcand, \msRV(\prefs)))).\]
\end{definition}

Thereby, we combine the semi-River diagram and the \ourPrim algorithm on that diagram, into an ordering that can be used with River.
The ordering \ospecial is a descending linear ordering over the edges in \mG, where for any set of edges with equal margin, the edges in $E(\rspt(\wcand,\msRV(\prefs)))$ are ordered before edges are not in $E(\rspt(\wcand,\msRV(\prefs)))$ by \Cref{def:ordering}.


\subsubsection*{The \rvPutCheck}
Using \ospecial we can now define our decision process for the winning set of River with PUT tiebreaking.
\begin{definition} \label{def:contructiveRVPUTcheck}
For a preference profile \prefs with the margin graph $\mG(\prefs)$, for any alternative $\wcand \in \alts$ we decide whether \wcand is a $\RVPUT$ winner of $\prefs$, by computing $\mRV(\ospecial, \prefs)$ and checking whether \wcand is the root of that River diagram. 
So we define \rvPutCheck(\wcand, \prefs) as the result of the following process:
\begin{enumerate}
    \item Compute \msRV(\prefs).
    \item Compute \rspt(\wcand, \msRV).
    \item Compute $\ospecial = \setOrdering(E(\mG), E(\rspt))$.
    \item Compute the winner $a$ of $\RV(\ospecial, \prefs)$.
    \item Return \texttt{true} if $\wcand = a$, return \texttt{false} otherwise.
\end{enumerate}
Formally, the \rvPutCheck(\wcand, \prefs) returns \texttt{true}, if and only if \[
    \wcand \in \RV(\setOrdering(E(\mG(\prefs)), E(\rspt(\wcand, \msRV(\prefs)))), \prefs).
\]
\end{definition}

\subsection*{Proof of Correctness and Computational Tractability}
Using this definition, we now show that this algorithm can be used to determine whether an alternative is a $\RVPUT$ winner for a preference profile. 
\begin{theorem} \label{thm:contructiveRVPUTcheckCorrect}
    The \rvPutCheck algorithm is correct.
\end{theorem}

To show correctness of the biimplication $\wcand \in \RVPUT(\prefs) \leftrightarrow \wcand \in \RV(\ospecial)$, we first assume that \wcand is a $\RVPUT$ winner and show that then \wcand is a winner of $\RV(\ospecial)$ in \Cref{thm:AtoC}. We do this by using the aforementioned properties to prove that the edges chosen by River given $\ospecial$ are equal to \Erspt. If \wcand is not a $\RVPUT$ winner, it cannot be the root of any River diagram, and the reverse direction in \Cref{thm:CtoA} is straightforward.
We prove that the \rvPutCheck does indeed only take polynomial time to compute in \Cref{thm:RvPutInP} because all steps of the process only take polynomial time.

The following theorem states the first direction. Its proof is of central importance to this thesis. 
\begin{theorem}
    \label{thm:AtoC}
    If $\wcand \in \alts$ is a winner of $\RVPUT(\prefs)$, then $\wcand$ is a winner of $\RV(\ospecial, \prefs)$ for any $\prefs \in \allPrefs$. 
\end{theorem}
In the following proof with preference profile \prefs, we write \rspt for $\rspt(\wcand,\msRV(\prefs))$ and \rvt for $\mRV(\ospecial, \prefs)$. We write \Ervti[i] for the edges of \rvti[i]. 

\setcounter{claim}{0}
\begin{proof}
 Let $\prefs \in \allPrefs$ be a preference profile with the margin graph \mG and let $\wcand\in \RVPUT(\prefs)$ be a River PUT winner i.e there exists an ordering $\lino \in \allDescOrders$ so that $\wcand\in \RV(\lino, \prefs)$. We prove that $\wcand \in \RV(\ospecial, \prefs)$. This is equivalent to \wcand being the root of the River diagram $\mRV(\ospecial, \prefs)$, which we show via the even stronger property that this River diagram does not only have the same root but is equal to the tree output by \ourPrim for \wcand on \msRV :
 \[ \rspt(\wcand, \msRV(\prefs)) = \mRV(\ospecial, \prefs).\]
To that end, we show that $\rspt(\wcand, \msRV)$ is a subset of $\mRV(\ospecial, \prefs)$ by induction. Having this property, the proof for the equality is straightforward.
Because $\rspt$ and $\rvt$ are rooted trees with the same amount of edges $\abs{\alts}-1$ it follows:
\begin{claim} \label{main:from_subset_follows_equality}
    From $\Erspt \subseteq \Ervt$ we can conclude $\Erspt = \Ervt $. 
\end{claim}

We now show that $\Erspt \subseteq \Ervt$ 
by proving the implication $e\in \Erspt \rightarrow e \in \Ervt$. For this, we use induction over the edges in $E(\mG)$ as ordered in \ospecial. Recall that this is the order in which River processes the edges to generate \Ervt.

The following claim is relevant for the entire proof.
\begin{claim} \label{main:edgesInSemiRiver}
    For any edge $e \in E(\mG)$, if $e\in\Ervt$ or $e\in \Erspt$, this implies $e \in \EsRV$.
\end{claim}
\begin{proof}
    This is because the edges of the semi-River diagram are a superset of all possible River diagrams by \Cref{pps:msRVsupersetMRV} and because \rspt is a subgraph of the semi-River diagram.  
\end{proof}

For our induction base, regard $\ospecial [1] \in E(\mG)$. Let $\ospecial [1] \in \Erspt$. Since $\ospecial [1]$ is the first edge processed by River, it is added to \Ervt.

For our induction hypothesis, we assume that $\ospecial [i'] \in \Erspt \rightarrow \ospecial [i'] \in \Ervt$ for $1\leq i'< i$ .
Note that from this follows: All edges with strictly higher margin than $\m{\ospecial [i]}$ that are in \Erspt are also in \Ervt.
From this also follows \Cref{main:greaterPathsAreIncluded}. 
\begin{claim} \label{main:greaterPathsAreIncluded}
    A path \PathP with a strength greater than $\m{\ospecial [i]}$ that is in \rspt is also in \rvt.
\end{claim}
\begin{proof}
    Let \PathP be a path in \rspt with a strength greater than $\m{\ospecial [i]}$.
    Let $e_\PathP \in \Erspt$ be any edge on \PathP. Since \PathP has a strength greater than $\m{\ospecial [i]}$, any $e_\PathP$ has a margin greater than $\m{\ospecial [i]}$ per definition of the $\str$. By the induction hypothesis $e_\PathP$ is in \Ervt.
    Since all edges on path \PathP are in \Ervt, it is in \rvt.
\end{proof}

To show the induction step, 
let $e=(x,y)=\ospecial [i] \in E(\mG)$. 
We assume
\begin{align} 
    e\in\Erspt \label{main:e_in_rspt}
\end{align}
since otherwise the implication is trivially true.
From this and \Cref{main:edgesInSemiRiver} follows 
\begin{align}
    e \in \EsRV. \label{main:e_in_semiRiver}
\end{align}

Recall that the River branching condition \iCond for $e$ states that $y$ already has an incoming edge in $\Ervti[i-1]$. The River cycle condition \iiCond for $e$ states that there is a path from $y$ to $x$ in $\Ervti[i-1]$. Recall from \Cref{river:intermediateForest} that the intermediate River diagram $\Ervti[i-1]$ contains no cycles and every alternative has at most one incoming edge. Refer back to \Cref{def:river} for the River process.
We show that if either the River branching condition or the River cycle condition are satisfied for $e$, this leads to a contradiction.
From this we will conclude that they must both be false for $e$ and thus $e$ is in the River diagram. 

\paragraph{Case \RN{1}} Assume towards contradiction that \iCond is satisfied for $e$.
This means that there is another incoming edge $\eiin$ of $y$ with $\m{\eiin} \geq \m{e}$ in the River diagram:
\begin{align} \label{main:caseI:assumption_eein}
    \eiin=(z,y) \in \Ervti[i-1].
\end{align}

We now show there is a stronger path from $y$ to $z$ in the semi-River diagram. 
From \Cref{main:edgesInSemiRiver} follows that 
\begin{align} \label{main:caseI:eein_in_semiRiver}
    \eiin \in \EsRV. 
\end{align}
Because \rspt is a tree with at most one incoming edge per alternative by \Cref{def:rspt} and $e\in \Erspt$ is an incoming edge to $y$, we deduce that
\begin{align}  \label{main:caseI:eein_not_rspt}
    \eiin \notin \Erspt.
\end{align}
Now we show that $\eiin$ cannot have equal margin to $e$, as this leads to a direct contradiction.
If $\m{\eiin} = \m{e}$, then by \Cref{def:ordering} the edge $\eiin$ is after $e=\ospecial[i]$ in \ospecial, because $e\in \Erspt$ by \Cref{main:e_in_rspt} but $\eiin \notin \Erspt$ by \Cref{main:caseI:eein_not_rspt}. This contradicts that $\eiin \in \Ervti[i-1]$ from \Cref{main:caseI:assumption_eein}.
Therefore, we assume 
\begin{align} \label{main:caseI:eein_higherthan_e}
    \m{\eiin} > \m{e}. 
\end{align}
Note that then $\eiin$ is before $e$ in the ordering \ospecial. 
Since $e$ and $\eiin$ are incoming edges of $y$ in the semi-River diagram \msRV and $\m{\eiin} > \m{e}$ by \Cref{main:e_in_semiRiver,main:caseI:eein_in_semiRiver,main:caseI:eein_higherthan_e}, it follows from \Cref{lma:semiRiverPathExistance} that there exists a path \Path{y}{z} in \msRV with a strength of at least $\m{\eiin}$:
\begin{align} \label{main:caseI:pyz_atleast_eiin_higher_e}
    \Path{y}{z} \text{ in } \msRV, \qquad \str(\Path{y}{z}) \geq \m{\eiin} > \m{e}.
\end{align}

We use this path with the recursive strongest path property of \rspt, whose edges \ospecial is based on, to prove that $\eiin \notin \ERVo$. 
Recall this property from \Cref{def:rspt} says that for $a,b\in\alts$, if $a$ precedes $b$ in the tree, then it contains a strongest path from $a$ to $b$ from $\msRV$. 
To that end, we show that $y$ precedes $z$ in \rspt.
For an illustration refer to \Cref{fig:main:caseI:allg}.
\begin{figure}[t] \centering
    \includegraphics[width = 0.4\linewidth]{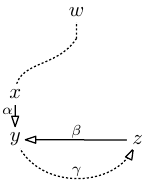}
    \caption{An illustration of \msRV for Case \RN{1}. The path from $\wcand$ to $y$ in \rspt goes through $e=(x,y)$. The edge $e$ has margin $\alpha$. There is the other incoming edge $\eiin =(z,y)$ with margin $\beta$. Note that $\alpha < \beta$. Because both edges are in the semi-River diagram, there exists a path $\Path{y}{z}$ form $y$ to $z$ with strength $\gamma \geq \beta$ in the semi-River diagram.}
    \label{fig:main:caseI:allg}
\end{figure}

We show this, by examining how the \ourPrim algorithm from \Cref{rspt:algorithm} creates \rspt. We distinguish the following cases: If it explores $z$ before $y$ this leads to a contradiction.
If it explores $y$ before $z$, we show that $y \prec z$ in \rspt follows.

Recall that \ourPrim explores the alternatives in the semi-River diagram \msRV by taking a maximum crossing edge, which is an edge with maximum margin among all edges connecting an already explored alternative in the set $S$ to an unexplored alternative in $\alts \setminus S$.


\begin{casesproof}
\item 
\emph{\ourPrim explores $z$ before $y$.} \quad
    This assumption contradicts that $e$ is in \Erspt. We show this by comparing $e=(x,y)$ to $\eiin=(z,y)$ in \EsRV.
    This is illustrated by \Cref{fig:main:caseI:noprec:equal:zbeforey}.
    \begin{figure}[t] \centering
        \includegraphics[width = 0.4\linewidth]{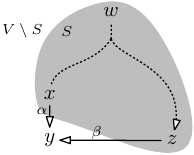}
        \caption{An illustration of \msRV for Case \RN{1}.1. We regard the point where $y\notin S$ and $z\in S$. The path from $\wcand$ to $y$ goes through the edge $e=(x,y)\in\Erspt$. The edge $e$ with margin $\alpha$ and the edge $\eiin=(z,y)$ with margin $\beta$ are crossing edges. Note that $\alpha<\beta$.}
        \label{fig:main:caseI:noprec:equal:zbeforey}
    \end{figure}

    If $e$ is added to \Erspt before $\eiin$ is a crossing edge, this would explore $y$ before $z$ and contradict our assumption.
    Also, $e$ cannot be added after $\eiin$ is a crossing edge as in that case $y$ is already explored, because $z$ is explored before $y$. 
    So the selection of $e$ into \Erspt can only happen when both $e$ and $\eiin$ are crossing edges, which contradicts that \ourPrim always selects the maximum crossing edge as $\m{e}<\m{\eiin}$ by \Cref{main:caseI:eein_higherthan_e}. 
    Thus if \ourPrim explores $z$ before $y$, the edge $e$ is not added to \Erspt which contradicts \Cref{main:e_in_rspt}. So the other case that \ourPrim explores $y$ first must be true.
\item
\emph{\ourPrim explores $y$ before $z$.} \quad
    In this case, we show that $y\prec z$ in \rspt follows from our assumptions.

    Note that $y$ is explored via $e$, because $e\in \Erspt$ by \Cref{main:e_in_rspt}.
    Let $E_{\text{cross}}$ be the set of crossing edges at the state of the execution of \ourPrim when $e$ is selected as a maximum crossing edge.
    Since $e$ has maximum margin in $E_{\text{cross}}$, it follows for any edge $e' \in E_{\text{cross}}$ that $\m{e'}\leq\m{e}$.
    Note that any edge $e_{yz}$ on \Path{y}{z} in \msRV has a higher margin than any $e'\in E_{\text{cross}}$: 
    \begin{align} \label{main:caseI:noprec:equalminweight:ybeforez:edges_on_path_yz_higher_any_pathwz}
    \m{e_{yz}} \geq \str(\Path{y}{z}) \geq \m{\eiin} > \m{e} \geq \m{e'},
    \end{align}
    by the definition of $\str$ and from \Cref{main:caseI:pyz_atleast_eiin_higher_e}.
    Upon the exploration of $y$ the first edge on this path is a crossing edge.
    This is illustrated by \Cref{fig:main:caseI:noprec:equal:ybeforez}.
    \begin{figure}[t] \centering
        \includegraphics[width = 0.4\linewidth]{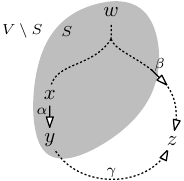}
        \caption{An illustration of \msRV for Case \RN{1}.2. We regard the point where $y\in S$ and $z\notin S$. The path from $\wcand$ to $y$ goes through the edge $e$ with margin $\alpha$ which is selected by \ourPrim to explore $y$. Thus its margin is at least as high as of any crossing edge $e'$ with margin $\beta$. The path \Path{y}{z} has the strength $\gamma$ with $\gamma>\beta$. }
        \label{fig:main:caseI:noprec:equal:ybeforez}
    \end{figure}
    
    After the exploration of $y$ via $e$ an edge $e'$ from $E_{\text{cross}}$ can only be selected once all edge of \Path{y}{z} are no longer crossing edges. This follows from \Cref{main:caseI:noprec:equalminweight:ybeforez:edges_on_path_yz_higher_any_pathwz}.
    We regard this point where all edge of \Path{y}{z} are no longer crossing edges. 
    \begin{casesproof}
    \item
        If all edges on \Path{y}{z} are no longer crossing, because all of them are selected into \Erspt, then $y \prec z$ in \rspt via this path.
    \item
        Otherwise, if not all edges on \Path{y}{z} are in \Erspt but none of them is still a crossing edge, it follows that there is some other path to $z$.
        Because the edges on this path to $z$ are selected instead of the edges on \Path{y}{z}, they have a margin higher than or equal to $\str(\Path{y}{z})$. This is not possible for any edge $e' \in E_{\text{cross}}$, so $y$ is the only starting point for such a path to $z$ with at least the strength of \Path{y}{z}.
        Thus $y \prec z$ in \rspt via this path.
    \end{casesproof}
     
    From this case distinction follows that $y \prec z$ in \rspt if $y$ is explored before $z$ by \ourPrim. 
\end{casesproof}
Because the other case leads to a contradiction, the \ourPrim explores $y$ before $z$ and thus $y \prec z$ in \rspt.
Then a path from $y$ to $z$ at least as strong as $\Path{y}{z}$ in \msRV exists in \rspt. This path is in $\rvti[i-1]$ by \Cref{main:greaterPathsAreIncluded}, and thus $\eiin \notin \Ervti[i-1]$ because otherwise it would close a cycle.
So we can conclude that \iCond is false for $e$.

\paragraph{Case \RN{2}} Assume towards contradiction that \iiCond is satisfied for $e=(x,y)$.
This means there is a path $\Path{y}{x}$ from $y$ to $x$ in \rvti[i-1]. 
 There must be some edge $\ebrm=(c,d) \in \msRV$ on $\Path{y}{x}$ with $\ebrm \notin \Erspt$, since by \Cref{main:e_in_rspt} edge $e\in\rspt$  and by \Cref{def:rspt} the \rspt is acyclic :
 \begin{align} \label{main:caseII:ebr_on_Pyx_not_in_Erspt}
    \ebrm=(c,d) \text{ on } \Path{y}{x} \qquad  \ebrm \notin \Erspt.
 \end{align}
 From the assumption that $\Path{y}{x}$ is in \rvti[i-1], follows that:
 \begin{align} \label{main:caseII:ebr_in_river}
    \ebrm \in \Ervti[i-1].
 \end{align}
 However, we show that this edge is not in \rvti[i-1], leading to a contradiction.
If $\m{\ebrm} = \m{e}$, then $e$ is before $\ebrm$ in \ospecial by \Cref{def:ordering,main:caseII:ebr_on_Pyx_not_in_Erspt,main:e_in_rspt}. Thus $\ebrm$ cannot be in \Ervti[i-1] and there is no $\Path{y}{x}$ in \rvti[i-1]. 
Therefore we can assume
\begin{align}  \label{main:caseII:ebr_higher_e}
    \m{\ebrm} > \m{e}.
\end{align}
Then $\ebrm$ is before $e$ in \ospecial.
For an illustration refer to \Cref{fig:main:caseII:allg}. 
\begin{figure}[t] \centering
    \includegraphics[width = 0.4\linewidth]{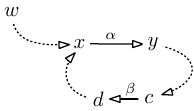}
    \caption{An illustration of \msRV for Case \RN{2}. The path from $\wcand$ to $y$ in \rspt goes through $e=(x,y) \in \Erspt$ which has a margin of $\alpha$. The path \Path{y}{x} is in the River-diagram $\rvti[i-1]$. There is an edge $\ebr=(c,d)$ on \Path{y}{x} with margin $\beta$. Note that $\beta >\alpha$.}
    \label{fig:main:caseII:allg}
\end{figure}

From $\ebrm=(c,d) \notin \Erspt$ also follows that either there is some other incoming edge $e'\in \msRV$ of $d$ that is in \Erspt instead of $\ebrm$ or there is no incoming edge of $d$ in \Erspt. 
In both cases $\ebrm$ is not in the River diagram.
\begin{casesproof}
    \item
    \emph{Let $e'$ be an incoming edge of $d$ in \Erspt.}
    If $\m{e'} > \m{e}$, then the edge $e'$ is in \Ervti[i-1] by our induction hypothesis. In that case, $\ebrm$ is not in \Ervti[i-1] because \rvt is a tree.
    So we assume
    \begin{align}\label{main:caseII:edashexists:edash_lower_ebr}
       e' \in \Erspt \qquad \m{e'} \leq \m{e}.
    \end{align}
    Since $\m{e'}<\m{\ebrm}$ it follows from \Cref{lma:semiRiverPathExistance} that there is a path \Path{d}{c} in \msRV, with $\m{\ebrm} \leq \str(\Path{d}{c})$.
    Note that from \Cref{main:caseII:ebr_higher_e} follows that $\m{e} < \str(\Path{d}{c})$.
    So we have:
    \begin{align}\label{main:caseII:edashexists:Pdc_atleast_ebr}
        \Path{d}{c} \text{ in } \msRV, \qquad \str(\Path{d}{c}) \geq \m{\ebrm} > \m{e} \geq \m{e'}.
    \end{align}
    For an illustration refer to \Cref{fig:main:caseII:edash}.
    \begin{figure}[t] \centering
        \includegraphics[width = 0.4\linewidth]{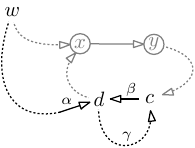}
        \caption{An illustration of \msRV for Case \RN{2}.1. The path from $\wcand$ to $d$ in \rspt goes through $e'$. There is the other incoming edge $\ebrm =(c,d)$. The margin $\alpha$ of $e'$ is smaller than the margin $\beta$ of $\ebrm$. Because both edges are in the semi-River diagram, there exists a path $\Path{d}{c}$ form $d$ to $c$ with strength $\gamma \geq \beta$ in the semi-River diagram.}
        \label{fig:main:caseII:edash}
    \end{figure}

    We now show that $d\prec c$ in \rspt.
    This is analogous to the argument in \textbf{Case \RN{1}}: For the alternatives $y,z \in \alts$ with edges $e \in \Erspt$ and $\eiin=(z,y) \in \msRV$ incoming to $y$ with $\m{e}<\m{\eiin}$ and with path \Path{y}{z} in \msRV with $\str(\Path{y}{z}) \geq \m{\eiin}$ from \Cref{main:e_in_rspt,main:caseI:eein_higherthan_e,main:caseI:pyz_atleast_eiin_higher_e} it follows that $y\prec z$ in \rspt.
    These are analogous to the alternatives $d,c \in \alts$ with edges $e' \in \Erspt$ and $\ebrm=(c,d) \in \msRV$ incoming to $d$ with $\m{e'}<\m{\ebrm}$ and with path \Path{d}{c} in \msRV with $\str(\Path{d}{c}) \geq \m{\ebrm}$ from \Cref{main:caseII:edashexists:edash_lower_ebr,main:caseII:ebr_on_Pyx_not_in_Erspt,main:caseII:edashexists:Pdc_atleast_ebr}. Also compare \Cref{fig:main:caseII:edash} to \Cref{fig:main:caseI:allg}.
    Thus, $d\prec c$ in \rspt and so there is a path from $d$ to $c$ in \rspt with a strength of at least the strength of \Path{d}{c} it follows that this path is in \rvti[i-1] by \Cref{main:greaterPathsAreIncluded}.
    So we can conclude that this case leads to a contradiction with $\ebrm$ being in $\Ervti[i-1]$ from \Cref{main:caseII:ebr_in_river} as this would create a cycle.
    \item
    \emph{Let $d$ have no incoming edges in \Erspt.}
    Then $d$ is the root $\wcand$ of \rspt. In this case, for all in edges $e_{\text{defeat}}=(l,d)$, there is a path from $d$ to $l$ with strength at least $\m{e_{\text{defeat}}}$ in the semi-River diagram \msRV, according to \Cref{lma:PUTwinnersBreakAllIncomingEdges}.
    Because $d$ is the root, these paths are in \rspt by \Cref{def:rspt}. 
    So the path from $d$ to $c$ in \rspt has a strength of at least $\ebrm=(c,d)$. By \Cref{main:caseII:ebr_higher_e} and \Cref{main:greaterPathsAreIncluded} this path is in \rvti[i-1]. This contradicts that $\ebrm$ is in $\Ervti[i-1]$ from \Cref{main:caseII:ebr_in_river} as this would create a cycle.
\end{casesproof}

So we can conclude that the assumption that $\ebrm$ is in the River diagram leads to a contradiction and thus $\ebrm \notin \Ervti[i-1]$. So there is no path $\Path{y}{x}$ in \rvti[i-1] and \iiCond is false for $e$.

\paragraph*{}
In \textbf{Case \RN{1}} and \textbf{Case \RN{2}} we show that \iCond and \iiCond cannot be satisfied for $e$ in \rvti[i-1] and thus $e \in \Ervti[i]$. This shows the induction step.

The induction proves that $\Erspt\subseteq \Ervt$. From \Cref{main:from_subset_follows_equality} now follows that $\Erspt=\Ervt$.
So if $\wcand \in \alts$ with $\wcand \in \RVPUT(\prefs)$ is the root of $\rspt(\wcand,\msRV(\prefs))$ it is also the root of $\mRV(\ospecial, \prefs)$. This means that $\wcand \in \RV(\ospecial, \prefs)$.
\end{proof}

We have shown that if the alternative \wcand is a $\RVPUT$ winner for a preference profile \prefs, then it is a $\RV(\ospecial, \prefs)$ winner.
Now we regard the case that \wcand is not a $\RVPUT$ winner.
Notice that for this direction, only the last stage of the process matters: Executing River. The output of the previous algorithms does not matter. 
\begin{theorem} \label{thm:CtoA} 
        For a preference profile \prefs, let the alternative $\wcand \in \alts$ not be a winner of $\RVPUT(\prefs)$. Then \wcand is not a winner of $\RV(\ospecial, \prefs)$, given our \ospecial.
\end{theorem}
\begin{proof}
    Let \prefs be a preference profile. Let $\wcand \in \alts$ with $\wcand \notin \RVPUT(\prefs)$.
    This means there is no possible tiebreak with which \wcand is 
    determined a winner by River. Since we only run River with a specific tiebreak $\ospecial$, the alternative \wcand can not be a winner of $\RV(\ospecial, \prefs)$.
\end{proof}

The previous two theorems can now be combined to show the correctness of \rvPutCheck : 
\begin{proof}[Proof of \Cref{thm:contructiveRVPUTcheckCorrect}]
    Let \prefs be a preference profile.
    If \wcand is a winner of $\RVPUT(\prefs)$ then \wcand is a winner of $\RV(\ospecial, \prefs)$. We have shown this in \Cref{thm:AtoC}.
    Then \rvPutCheck(\prefs) returns \texttt{true} for \wcand by \Cref{def:contructiveRVPUTcheck}. 
    If \wcand is not a winner of $\RVPUT(\prefs)$ then \wcand is not a winner of $\RV(\ospecial, \prefs)$. We have shown this in \Cref{thm:CtoA}. Then \rvPutCheck(\prefs) returns \texttt{false} for \wcand by \Cref{def:contructiveRVPUTcheck}.
    Thus the \rvPutCheck algorithm is correct. 
\end{proof}

Finally:
\begin{theorem}\label{thm:RvPutInP}
    We can decide $\wcand\in \RVPUT(\prefs)$ in polynomial time 
    for any preference profile $\prefs \in \allPrefs$ and $\wcand \in \alts$ .
\end{theorem}
\begin{proof}
    By \Cref{thm:contructiveRVPUTcheckCorrect} the \rvPutCheck(\prefs) decides whether $\wcand\in \RVPUT(\prefs)$.
    All steps in this process have a runtime polynomial in $\abs{\alts}$ and $\abs{E(\mG)}$.
    This follows from \Cref{thm:srvInP}, \Cref{obs:rsptInP}, \Cref{pps:oInP}, and \Cref{lma:rvInP},
    as these steps in the process are sequential and have a worst-case runtime polynomial in $\abs{\alts}$. Thus the runtime is dominated by creating the semi-River diagram, which takes $\bigO{\abs{\alts}^5}$.
    So $\RVPUT(\prefs)$ can be decided in polynomial time.  

\end{proof}
Thus, we can also decide the set of River PUT winners for a preference profile \prefs in polynomial time because \[\RVPUT(\prefs) = \msset{\wcand\in \alts \mid \rvPutCheck(\wcand, \prefs) = \texttt{true}}.\]


    \ifthenelse{\boolean{finishedLook}}{ }{
    \section{Computing $\RVPUT$} \label{sct:winnerSet}
Using the constructive $\RVPUT$ check, we can determine weither an alternate $\wcand \in \alts$ is a $\RVPUT$ winner.

To compute the winning set we can
\begin{enumerate}[label=\alph*)]
    \item Use the constructive $\RVPUT$ check for every alternative in \alts.
    \item Compute additional information from the semi-River diagram, to directly determine the winning set.
\end{enumerate}

For b) we define the following edge types:
\subsection*{Edge Types}
These edge types classify edges in the margin graph, based on their behaviour over all Tie-Breaks.
\begin{definition}\label{def:impossibleEdge}
    We define an edge $e \in \mG$ as \impossible, if is in no River diagram for any Tie-Break. Formally: $e$ is \impossible iff $\forall \lino \in \allDescOrders: e\not \in \mRV(\lino)$.
\end{definition}
These are the edges we want to exclude from the semi-River diagram.

\begin{definition}\label{def:possibleEdge}
    We define an edge $e\in \mG$ as \possible, if there is in a River diagram for some Tie-Break. Formally: $e$ is \possible\ iff $\exists \lino \in \allDescOrders: e \in \mRV(\lino)$.
\end{definition}
These are the edges we want to include in the semi-River diagram.

\begin{corollary}
    The \Cref{def:strongSemiRiverProcess} is equivalent to: For every edge $e\in \mG$, we include $e$ in \msRV, iff $e$ is \possible (and exclude $e$ iff $e$ is \impossible).
\end{corollary}
\begin{proof}
    This follows directly from \Cref*{pps:msRVsubsetMRV} and \Cref*{pps:msRVsupersetMRV}.
\end{proof}

\subsection*{More Edge Types for \possible edges}
Among the \possible edges, we can draw further distinctions. First we differentiate \certain from \uncertain edges, both of which we have in the semi-River diagram.

\begin{definition}\label{def:certainEdge}
    We define an edge $e \in \mG$ as \certain, if it is in the River diagram for every Tie-Break. Formally: $e$ is \certain iff $\forall \lino \in \allDescOrders: e\in \mRV(\lino)$.
\end{definition}
\certain edges are single incoming edge in a the semi-River diagram and there is no path with higher or equal weight back
To give an intuition, one such edge could be the a unique highest edge of the margin graph, for which \iCond and \iiCond are trivally false, for any linear descending order.

\begin{definition} \label{def:uncertainEdge}
    We define an edge $e\in\mG$ as \uncertain, if it is in the River diagram for some Tie-Break and not in the River diagram for some other Tie-Break. Formally: $e$ is \uncertain iff $\exists \lino\in \allDescOrders: (e \in \mRV(\lino)) \land \exists \lino' \in \allDescOrders: (e \not \in \mRV(\lino'))$.
\end{definition}

To better understand the \uncertain edges, we divide them based on why they might be excluded.
\begin{definition}\label{def:breakableEdge}
    We define an edge $e\in\mG$ as \breakable, if it is \uncertain and there is a River diagram for a Tie-Break where it is excluded because of \iiCond, that is, it would close a cycle under that Tie-Break. Formally: $e=(x,y)$ is \breakable iff $\exists \lino\in \allDescOrders: e \not \in \mRV(\lino) \land \exists y \rightsquigarrow_p x \text{in} \mRVi[\lino^-1(e) +1](\lino)$ and $e$ is \uncertain.
\end{definition} 
We will later show that alternatives which only have \breakable incoming edges are $\RVPUT$ winners.

\begin{definition}\label{def:choosableEdge}
    We define an edge $e\in\mG$ as \choosable, if it is \uncertain and there is a River diagram for a Tie-Break where it is excluded because of \iCond, that is, it would we the second incoming edge for an alternative and it is not \breakable. Formally: $e= (x,y)$ is \choosable iff $\exists \lino\in \allDescOrders: e \not \in \mRV(\lino) \land \exists e'\in \ine(e') \text{in} \mRVi[\lino^-1(e) +1](\lino)$ and $e$ is \uncertain and $e$ is not \breakable.
\end{definition}
Observe that this definition is the same as the edge being \uncertain but not \breakable, because if \iiCond is always false, but there is some Tie-Break where the edge is included and some where it is not than \iCond must be true for some cases.

\subsection*{Determine the $\RVPUT$ winning set using \breakable edges}
For a semi-River diagram \msRV(\prefs), we can can compute the set of \breakable edges:
For each edge $e=(y,x)\in \EsRV$, $e$ is \breakable iff there exists a path $P_{y,x}$ with $\str(P_{y,x})\geq\m{e}$ in \msRV. This is correct, because all edges in \msRV are \possible and if such a path exists they are not \certain and not \choosable.
The set of $\RVPUT$ winners are the alternatives for which all incoming edges are \breakable.
This is true, because all \breakable edges incoming to one alternative can all be excluded in the same River diagram. They are excluded because of \iCond. They can not be part of each others excluding paths. These excluding paths cannot contradict each other.
If $P$ requires a different incoming edge than $P'$  for some alternive in both, $P$ up to this alternative would have a differnt strength form $P'$ up to this alternative. There is a universe where to path with higher strength is chosen for connecting the alternative with \breakable edges and the common alternative and both the \breakable edges of $P$ and $P'$ are broken. 
If some subpath of $P$ is a path that excludes some edge of $P'$, this must contradict the fact that both start in the same alternative and end in one of its neighbors.

\subsection*{state purgatory}
Here we define the states the edges of the semi-River diagram can have in terms of how we compute than. We will later show the equivalence of these definitions with those of the previous subsection.
\begin{itemize}
    \item For the semi-River algorithm we define \textit{certain} edges, as edges $e=(x,y) \in \msRV$ with $\sset{e} = \ine_{\msRV}(y)$ and $\not \exists y \rightsquigarrow_p x : \str(p) \geq \str(p)$ in $\msRV$.
    This is the same as edge $e=(x,y) \in \msRVm{\m{e}}$ with $\sset{e} = \ine_{\msRVm[\m{e}]}(y)$ and $\not \exists y \rightsquigarrow_p x$ in $\msRVm[\m{e}]$.
    \item We define \textit{breakable} edges, as edges $e=(x,y) \in \msRV$ with $\exists y \rightsquigarrow_p x : \str(p) \geq \str(p)$ in \msRV.
    This is the same as edges $e=(x,y) \in \msRVm[\m{e}]$ with $\exists y \rightsquigarrow_p x$ in $\msRVm[\m{e}]$.
    \item We define \textit{chooseable} edges, as edges $e=(x,y) \in msRV$ which are not breakable ($\not \exists y \rightsquigarrow_p x : \str(p) \geq \str(p)$ in \msRV).
    This is the same as edges $e=(x,y) \in msRVm[\m{e}]$ which are not breakable ($\not \exists y \rightsquigarrow_p x$ in $\msRVm[\m{e}]$). All incoming edges to $y$ in $\msRVm[\m{e}]$, are breakable or choosable.
    \item We define \textit{impossible} edges, as edges that are not possible. \texttt{This def and the negation of possible edges, is the same as the semi River algo}
\end{itemize}
We call an edge $e=(x,y)$ \textit{possible} (not as a edge type, but as a helpful def for edge types), if  (not impossible; no certain cycle and no certain or chooseable higher edge) there is no strictly higher in edge that is certain or choosable and there is no path from $y$ to $x$ with edges that are certain or chooseable.

    }

    \makeatletter
        \def\toclevel@chapter{-1}
        \def\toclevel@section{0}
    \makeatother

    \chapter{Conclusions \& Outlook}
    \section*{Conclusions}
In this thesis, we proved that River with PUT can be computed in polynomial time with the \rvPutCheck algorithm.
This algorithm first creates a semi-River diagram, encoding important information about the River process, then creates an \rsptf for every alternative. It uses that tree to create a tiebreak that can be used to decide if an alternative is a winner of River PUT. 
For every winning candidate, this creates a rebutting tree as a result of the River run, which provides an easy-to-read certificate for its immunity.

For uniquely weighted preference profiles, we showed that River has a runtime of $\bigO{n^2 \log n}$ which is dominated by the sorting of the edges. The well-known Ranked Pairs methods has to check for cycles in a graph with fewer constraints, than the River diagram. Thus the naive runtime of Ranked pairs is in $\bigO{n^4}$, where this check is performed with a BFS for every edge.

The Ranked Pairs method is currently widely used, for example by Openassistant. There human curators rank replies according to their preference. The different rankings are merged with Ranked Pairs \cite{10.5555/3666122.3668186}.
In such a large dataset ties are plausible to occur. River with PUT could be applied in this context.

\section*{Future work}
We only showed that the semi-River diagram is a superset of all possible River diagrams. We conjecture the semi-River process can be adjusted so that the set of edges in the semi-River diagram and the set of edges in a River diagram under some tiebreak are equal. This is currently not so, because there are cases where at least one of the conditions is always satisfied for an edge but not the same conditions for every tiebreak.

The runtime of the \rvPutCheck given in this thesis is a naive runtime to show membership in P. This could likely be improved by optimizing the semi-River process. One example is shortcutting the later \rvPutCheck stages by investigating if the $\RVPUT$ winners can be directly read from the semi-River diagram. 
A drawback of $\RVPUT$ is that it is not resolute. Future work could investigate how to remedy this, by making some selection from the set of $\RVPUT$ winners.

Further investigations towards the properties of $\RVPUT$ seem most promising since the very similar Ranked Pairs method with PUT has neutrality.
More extensive properties already shown for River by \textcite{doringFlowingFairnessRiver2024} only hold for tiebreakers with certain characteristics and may also hold with PUT. 

It is surprising that River with PUT is in P, while Ranked Pairs is NP-hard, as they are computed in the same way, except that River has the additional branching condition. A direction of future work could be to further analyze this difference.

    \pagestyle{plain}

    \renewcommand*{\bibfont}{\small}
    \printbibheading
    \addcontentsline{toc}{chapter}{Bibliography}
    \printbibliography[heading = none]

    \addchap{Declaration of Authorship}
    I hereby declare that this thesis is my own unaided work. All direct or indirect sources used are acknowledged as references.\\[6 ex]

I used the following tools to improve the qualitity of the writing:
\begin{itemize}
    \item I used Writefull (www.writefull.com) for Overleaf and Grammarly (www.grammarly.com) to check for language mistakes.
    \item I used DeepL (www.deepl.com) to find better wording.
\end{itemize}


\end{document}